\documentclass[submission,copyright,creativecommons]{eptcs}
\providecommand{\event}{TARK 2019} 
\usepackage{breakurl}             
\usepackage{underscore}           
\usepackage{amsmath, amsfonts, amssymb, amsthm}
\usepackage{color}
\usepackage{diagbox}
\usepackage[mdyyyy]{datetime}
\usepackage{hyperref}
\hypersetup{
  colorlinks=true,
  linkcolor=blue,
  filecolor=blue,      
  urlcolor=blue,
 }
\usepackage{tikz,pgf} 
\usetikzlibrary{matrix, arrows, positioning, graphs, trees, shapes, arrows.meta} 
\tikzset{
  world/.style={circle,draw=black, minimum size=.4cm,thick},
  point/.style={circle,draw,inner sep=0.5mm,fill=black}
  ->,>={Stealth[width=1mm,length=1mm]},shorten >=1pt,shorten <=1pt,auto,node distance=1cm,semithick,
  sibling distance= 10em}
\usepackage{multirow,array}
\usepackage{pgffor}
\usepackage{enumitem}
\usepackage{subcaption}

\newtheorem{theorem}{Theorem}[section]
\newtheorem{definition}[theorem]{Definition}
\newtheorem{proposition}[theorem]{Proposition}
\newtheorem{lemma}[theorem]{Lemma}
\newtheorem{example}[theorem]{Example}

\newtheorem{notation}[theorem]{Notation}

\renewcommand{\phi}{\varphi}

\newcommand{\NN}{\mathbb{N}}

\newcommand{\bisim}{\leftrightarrows}
\newcommand\precdot{\mathrel{\ooalign{$\prec$\cr
  \hidewidth\raise0.0ex\hbox{$\cdot\mkern0.7mu$}\cr}}}
\newcommand{\imsub}{\precdot}
\newcommand{\len}{\textrm{len}}
\usepackage{latexsym}
\usepackage{verbatim}
\usepackage{mathtools}

\title{When Do Introspection Axioms Matter for Multi-Agent Epistemic Reasoning?}
\author{Yifeng Ding
\institute{University of California, Berkeley}
\email{yf.ding@berkeley.edu}
\and
Wesley H. Holliday
\institute{University of California, Berkeley}
\email{wesholliday@berkeley.edu}
\and
Cedegao Zhang
\institute{University of California, Berkeley}
\email{cedzhang@berkeley.edu}
}
\def\titlerunning{When Do Introspection Axioms Matter for Multi-Agent Epistemic Reasoning?}
\def\authorrunning{Y. Ding, W. H. Holliday, \& C. Zhang}
\begin{document}
\maketitle

\begin{abstract}
The early literature on epistemic logic in philosophy focused on reasoning about the knowledge or belief of a single agent, especially on controversies about  ``introspection axioms'' such as the $\mathsf{4}$ and $\mathsf{5}$ axioms. By contrast, the later literature on epistemic logic in computer science and game theory has focused on multi-agent epistemic reasoning, with the single-agent $\mathsf{4}$ and $\mathsf{5}$ axioms largely taken for granted. In the relevant multi-agent scenarios, it is often important to reason about \textit{what agent A believes about what agent B believes about what agent A believes}; but it is rarely important to reason just about \textit{what agent A believes about what agent A believes}. This raises the question of the extent to which single-agent introspection axioms actually matter for multi-agent epistemic reasoning. In this paper, we formalize and answer this question. To formalize the question, we first define a set of multi-agent formulas that we call \textit{agent-alternating formulas}, including formulas like $\Box_a \Box_b \Box_a p$ but not formulas like $\Box_a \Box_a p$. We then prove, for the case of belief, that if one starts with multi-agent \textsf{K} or \textsf{KD}, then adding both the $\mathsf{4}$ and $\mathsf{5}$ axioms (or adding the \textsf{B} axiom) does not allow the derivation of any new agent-alternating formulas---in this sense, introspection axioms do not matter. By contrast, we show that such conservativity results fail for knowledge and multi-agent \textsf{KT}, though they hold with respect to a smaller class of \textit{agent-nonrepeating formulas}.
\end{abstract}

\section{Introduction}
The classic early works on epistemic logic in philosophy by Hintikka \cite{Hintikka1965} and Lenzen \cite{Lenzen1978} focused on the logic of knowledge and belief for a single agent,\footnote{Only \S\S~4.1-4.6 and \S~4.13 of \cite{Hintikka1965} and pp.~59, 66, and 70 of \cite{Lenzen1978} contain discussion of multi-agent formulas.} especially on controversies about ``introspection axioms": for example, if an agent knows $p$, does she know that she knows $p$ (formalized by the $\mathsf{4}$ axiom of modal logic, $K_ap\to K_aK_ap$)? If an agent does not know $p$, does she know that she does not know $p$ (formalized by the $\mathsf{5}$ axiom of modal logic, $\neg K_ap\to K_a\neg K_ap$)? By contrast, the later literature on epistemic logic in computer science  (e.g., \cite{Meyer1995,Fagin2003}) and game theory (e.g., \cite{Aumann1999}) focused on \textit{multi-agent} epistemic reasoning, especially as required for coordination between agents or strategic reasoning against opponents. In this literature, the single-agent introspection principles formalized by the $\mathsf{4}$ and $\mathsf{5}$ axioms are largely taken for granted (for exceptions, see, e.g., \cite{Vardi1985,Lamarre1994,Kaneko2002}). In the relevant multi-agent scenarios, it is often important to reason about \textit{what agent A believes about what agent B believes about what agent A believes} ($B_aB_bB_ap$); but it is rarely important to reason just about \textit{what agent A believes about what agent A believes} ($B_aB_ap$). Consider the following famous examples of multi-agent epistemic reasoning.

\paragraph{Muddy children} We assume familiarity with the 3-agent Muddy Children puzzle where two children have mud on their foreheads (see, e.g., \S~1.1 of \cite{Fagin2003}). The following is a derivation in the bimodal version of the minimal normal modal logic \textsf{K} showing how one of the muddy children comes to realize that she is muddy.\footnote{`PL' stands for propositional logic, `Nec' stands for the necessitation rule, and `RM' stands for the monotonicity rule that if $\varphi\to\psi$ is a theorem, then so is $\Box_i\varphi\to\Box_i\psi$. Note that in the derivation, RM is only applied to theorems of the logic. For example, to obtain $(4)$, RM is applied to the theorem $(\varphi\to\psi)\to(\neg\psi\to\neg\varphi)$ where $\varphi:=\Box_2(\neg m_1\wedge\neg m_3)$ and $\psi:=\Box_2m_2$.} Note that (i) no introspection axioms are used, and in fact (ii) modalities occur only ``alternatingly,'' in the sense that no occurrence of a modality for an agent $i$ has scope over another occurrence of a modality for $i$ without an intervening occurrence of some modality for an agent $j\neq i$.
\begin{enumerate}[label=(\alph*)]
    \item $\Box_1 \Box_2 ((\neg m_1\wedge \neg m_3)\to m_2)$ \hfill(assumption: 1 knows that 2 knows that at least one child is muddy)
    \item $\Box_1 \Box_2 \neg m_3$ \hfill(assumption: 1 knows that 2 can see 3, who is not muddy)
    \item $\Box_1 (\lnot m_1 \to \Box_2 \lnot m_1)$ \hfill(assumption: 1 knows that 2 can see 1)
    \item $\Box_1 \lnot \Box_2 m_2$ \hfill(assumption: 1 knows that 2 did not step forward after the parent's first question)
\end{enumerate}
\begin{enumerate}
    \item $\Box_2 ((\neg m_1\wedge \neg m_3)\to m_2) \to (\Box_2 (\neg m_1\wedge \neg m_3)\to \Box_2 m_2)$ \null \hfill (K axiom)
    \item $\Box_1 \Box_2 ((\neg m_1\wedge \neg m_3)\to m_2) \to \Box_1(\Box_2 (\neg m_1\wedge \neg m_3)\to \Box_2 m_2) $ \hfill (from (1) by RM)
    \item $\Box_1(\Box_2 (\neg m_1\wedge \neg m_3)\to \Box_2 m_2)$ \hfill (from (a) and (2) by PL)
    \item $\Box_1(\neg \Box_2 m_2\to \neg\Box_2 (\neg m_1\wedge \neg m_3))$ \hfill (from (3) using PL and RM)
    \item $\Box_1\neg\Box_2 (\neg m_1\wedge \neg m_3)$ \hfill (from (d) and (4) by K and PL)
    \item  $\neg \Box_2 (\neg m_1\wedge \neg m_3)\to\neg (\Box_2 \neg m_1\wedge \Box_2 \neg m_3)$ \hfill (theorem of \textsf{K})
    \item $\Box_1\neg \Box_2 (\neg m_1\wedge \neg m_3)\to \Box_1\neg (\Box_2 \neg m_1\wedge \Box_2 \neg m_3)$ \hfill (from (6) by RM)
    \item $\Box_1\neg (\Box_2 \neg m_1\wedge \Box_2 \neg m_3)$ \hfill (from (5) and (7) by PL)
    \item $\Box_1 \neg \Box_2 \neg m_1$ \hfill (from (b) and (8) using PL, Nec, and K)
    \item $\Box_1 m_1$ \hfill (from (c) and (9) using PL, Nec, and K)
\end{enumerate}

\tikzset{
    solid node/.style={circle,draw,inner sep=1.5,fill=black},
    hollow node/.style={circle,draw,inner sep=1.5}
}
\paragraph{Backward induction} We assume familiarity with the classic backward induction reasoning in extensive form games (see, e.g., \cite[\S~6.2]{Osborne1994}). In \cite{Vilks1999}, Vilks provides a syntactical derivation of backwards induction in the bimodal version of the modal logic \textsf{KT}, which we reproduce below. Again note that (i) no introspection axioms are used, and in fact (ii) modalities occur only ``alternatingly'' as above.
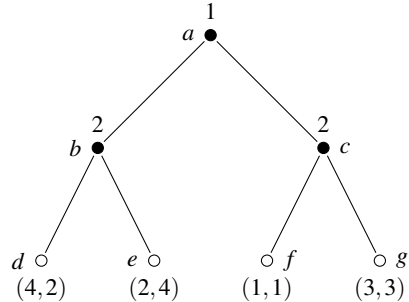
\begin{center}
\begin{tikzpicture}[scale=1.5,font=\footnotesize]
    \tikzstyle{level 1}=[level distance=10mm,sibling distance=20mm]
    \tikzstyle{level 2}=[level distance=10mm,sibling distance=10mm]
  \node(0)[solid node,label=above:{$1$},label=left:{$a$}]{}
    child{node(1)[solid node,label=above:{$2$},label=left:{$b$}]{}
      child{node[hollow node,label=below:{$(4,2)$},label=left:{$d$}]{} edge from parent node[left]{}}
      child{node[hollow node,label=below:{$(2,4)$},label=left:{$e$}]{} edge from parent node[right]{}}
      edge from parent node[left,xshift=-3]{}
    }
    child{node(2)[solid node,label=above:{$2$},label=right:{$c$}]{}
      child{node[hollow node,label=below:{$(1,1)$},label=right:{$f$}]{} edge from parent node[left]{}}
      child{node[hollow node,label=below:{$(3,3)$},label=right:{$g$}]{} edge from parent node[right]{}}
      edge from parent node[right,xshift=3]{}
    };
\end{tikzpicture}   
\end{center}
\begin{itemize}
    \item $p_1 \coloneqq ab \wedge \neg ac \wedge bd \wedge \neg be \wedge \neg cf \wedge \neg cg $ \hfill(both play left)
    \item $p_2 \coloneqq ab \wedge \neg ac \wedge \neg bd \wedge be \wedge \neg cf \wedge \neg cg $ \hfill(1 play left, 2 play right)
    \item $p_3 \coloneqq \neg ab \wedge ac \wedge \neg bd \wedge \neg be \wedge cf \wedge \neg cg $ \hfill(1 play right, 2 play left)
    \item $p_4 \coloneqq \neg ab \wedge ac \wedge \neg bd \wedge \neg be \wedge \neg cf \wedge cg $ \hfill(both play right)
    \item $q \coloneqq d >_1 e \wedge d >_1 f \wedge d >_1 g \wedge e >_1 f \wedge g >_1 e \wedge g >_1 f \wedge e >_2 d \wedge d >_2 f \wedge g >_2 d \wedge e >_2 f \wedge e >_2 g \wedge g >_2 f$ \hfill(players' preferences)
    \item $G \coloneqq (p_1 \vee p_2 \vee p_3 \vee p_4) \wedge q$ \hfill(description of the game)
\end{itemize}
\begin{enumerate}[label=(\alph*)]
    \item $\Box_1 G$ \hfill(assumption: 1 knows the game)
    \item $\Box_1((bd \vee be) \rightarrow \Diamond_2 be)$ (assumption: 1 knows that if 2 is at $b$ then 2 considers the move $be$ possible)
    \item $\Box_1((cf \vee cg) \rightarrow \Diamond_2 cg)$ \hfill(assumption: similar to (b))
    \item $\Box_2 ((ab \vee ac) \rightarrow \Diamond_1 ac)$ \hfill(assumption: similar to (b))
    \item $\Box_1((e >_2 d \wedge \Diamond_2 be) \rightarrow \neg bd)$ \hfill(assumption: follows from assuming 1 knows that 2 is rational)
    \item $\Box_1((g >_2 f \wedge \Diamond_2 cg) \rightarrow \neg cf)$ \hfill(assumption: similar to (e))
    \item $\big( \Box_1 (ab \leftrightarrow be) \wedge \Box_1 (ac \leftrightarrow cg) \wedge g >_1 e \wedge \Diamond_1 ac \big) \rightarrow \neg ab$ (assumption: follows from 1 being rational)
\end{enumerate}
\begin{enumerate}
    \item $\Box_1 (ab \leftrightarrow be)$ \hfill(from (a), (b), and (e) using PL, Nec, and K)    
    \item $\Box_1 (ac \leftrightarrow cg)$ \hfill(from (a), (c), and (f) using PL, Nec, and K)
    \item $(ab \vee ac) \rightarrow \Diamond_1 ac$ \hfill(from (d) by T)
    \item $G$ \hfill(from (a) by T)
    \item $ab \vee ac$ \hfill(from (4) by PL)
    \item $\Diamond_1 ac$ \hfill(from (3) and (5) by PL)
    \item $\neg ab$ \hfill(from (g), (1), (2), (4), and (6) by PL)
    \item $ac$ \hfill(from (5) and (7) by PL)
    \item $ac \leftrightarrow cg$ \hfill(from (2) by T)
    \item $ac \wedge cg$ \hfill(from (8) and (9) by PL)
\end{enumerate}

In general, in typical strategic form games a player needs to reason about the beliefs of her opponents, as which action is best for her depends on her opponents' actions, which in turn depend on their beliefs. On the other hand, reasoning about one's own beliefs seems unnecessary, as the dependencies just mentioned seem to be tight: which action is the best for a player depends on what her opponents' actions are alone, which in turn depend on their beliefs over what their opponents' actions are alone. We can then iterate this reasoning, and it seems there is no place for reasoning about one's own beliefs. In Appendix~\ref{sec:alternating-rationality} we provide a formalization of this idea using Kripke models of games in the style of \cite{stalnaker1994evaluation} and \cite{bonanno2002modal}, where only formulas with no modality scoping immediately over a modality of the same agent are used to ensure that rationalizable strategies are played.\footnote{We are not arguing that introspection assumptions never matter in multi-agent epistemic reasoning. For example, it is shown in \cite{Geanakoplos1989,Lederman2015} that Aumann's \cite{Aumann1976} theorem on agreeing to disagree fails without the assumption of positive introspection.}

These considerations raise the question of the extent to which single-agent introspection axioms actually matter for multi-agent epistemic reasoning. In particular, as motivated by the above examples, we can ask: in situations where the agents and also the analyst only need to reason about formulas where modalities occur only alternatingly, would the commonly debated introspection axioms still matter, in the sense that assuming them allows us to derive more conclusions?

This question has indeed been partially investigated previously, though motivated not by the question of whether introspection axioms may in practice be ``irrelevant'' but rather by the goal of devising efficient reasoning algorithms for the system $\mathsf{K45}$. In \cite{lakemeyer2012efficient}, it is explicitly stated (Lemma 5) that when restricted to the fragment of the multi-agent language in which modalities occur only in the agent-alternating way, $\mathsf{K}$ and $\mathsf{K45}$ derive the same set of theorems.\footnote{The authors refer to \cite{halpern2001multi} for the proof of this lemma, though we are unable to locate an explicit proof there.} This facilitates reasoning in $\mathsf{K45}$ since it is also known that every formula is provably equivalent in $\mathsf{K45}$ to an agent-alternating formula,\footnote{In Appendix \ref{sec:expressivity}, we show the semantic counterpart of this proposition and further show that $\mathsf{4}$ and $\mathsf{5}$ are in a sense necessary. See also Theorem 1 of \cite{parikh1992levels} for an early precursor of this result.} which is then derivable in $\mathsf{K45}$ iff it is derivable $\mathsf{K}$, making the efficient methods of deciding theoremhood in $\mathsf{K}$ applicable to $\mathsf{K45}$. Subsequently, the idea of agent-alternating formulas was also used in the axiomatization of refinement quantification logics \cite{hales2012refinement,hales2016quantifying} and in  epistemic planning \cite{huang2018general,liu2018multi,Fang2019compile}. 

In this paper, we study the question more systematically. In \S~\ref{sec:agent-alternating-formulas}, we provide multiple ways to define the \textit{agent-alternating formulas}, which include formulas like $\Box_a(\Box_b p \land \Box_b  \Box_a q)$ but not $\Box_a(\Box_b p \land \Box_a q)$. In \S~\ref{sec:alt-collapse}, we first provide a bisimulation notion for the fragment of agent-alternating formulas and then use it to completely chart the relationships of the  modal logics in the well-known ``Modal Logic Cube'' when restricted to the fragment of agent-alternating formulas. We prove that if one starts with multi-agent \textsf{K} or \textsf{KD}, then adding both the $\mathsf{4}$ and $\mathsf{5}$ axioms (or adding the $\mathsf{B}$ axiom) does not allow the derivation of any new agent-alternating formula---in this sense, introspection axioms do not matter. By contrast, we show that such conservativity results fail for knowledge and multi-agent \textsf{KT}, though they hold with respect to a smaller class of \textit{agent-nonrepeating formulas} introduced in \S~\ref{sec:nonrepeating}. In \S~\ref{sec:commonbelief}, we report on preliminary investigations of how these results are affected in the presence of a \textit{common belief} operator in the language. Finally, we conclude in \S~\ref{sec:discussion} with some directions for future research.

\section{Agent-Alternating Formulas}
\label{sec:agent-alternating-formulas}
Fix a set $A$ of agents with $|A| \ge 2$ and a countably infinite set $\mathsf{Prop}$ of proposition letters.

\begin{definition} The language of multi-agent epistemic logic is defined inductively by
    $$
    \mathcal{L} \ni \varphi ::= 
    p \mid 
    \lnot \varphi \mid 
    (\varphi \land \varphi) \mid 
    \Box_a\varphi
    $$
    where $p\in \mathsf{Prop}$ and $a\in A$. Connectives $\to$, $\lor$, and $\Diamond_a$ are abbreviations as usual. 
\end{definition}
  
We adopt the standard definition of when one formula is a \textit{subformula} of another.

\begin{notation} \normalfont{For $\varphi,\psi\in\mathcal{L}$, let $\varphi\preccurlyeq\psi$ indicate that $\varphi$ is a subformula of $\psi$ and $\varphi \prec \psi$ that $\varphi$ is a proper subformula of $\psi$.}
\end{notation}

Intuitively, agent-alternating formulas are those formulas in which an operator $\Box_a$ does not immediately scope over another operator $\Box_a$ of the same agent $a$. We now offer two ways to precisely capture this intuition, one using immediate subformulas and occurrences, and one using simultaneous induction. 
  
\begin{definition}\normalfont{
    For $\alpha, \beta \in \mathcal{L}$, we say $\alpha$ is an
    \textit{immediate subformula} of $\beta$, and write $\alpha \imsub \beta$, if
    $\beta$ is either $\lnot \alpha$, or $(\alpha \land \gamma)$ for some
    $\gamma \in \mathcal{L}$, or $(\gamma \land \alpha)$ for some $\gamma \in
    \mathcal{L}$, or $\Box_a\alpha$ for some $a \in A$. Note that the
    reflexive and transitive closure of $\imsub$ is precisely $\preccurlyeq$.
    
    For any $\varphi \in \mathcal{L}$, an \textit{occurrence type} $O$ of
    $\varphi$ is a finite sequence $\langle O_1, O_2, \cdots, O_{\textrm{len}(O)}
    \rangle$ of formulas in $\mathcal{L}$ such that $O_{\textrm{len}(O)} = \varphi$ and
    for each $i$ between $1$ and $\textrm{len}(O) - 1$, $O_i \imsub O_{i+1}$. Let
    $OC(\varphi)$ be the set of occurrence types of $\varphi$ and $\le$  the
    prefix-extension relation: $O \le O'$ iff $O'$ is a suffix of $O$. It is
    then easy to see that $\langle OC(\varphi), \le \rangle$ is a
    (downward-growing) tree.
    
    We call an occurrence type $O$ of $\varphi$ with $O_1 = \alpha$ an
    \textit{$\alpha$-occurrence} of $\varphi$. If this $\alpha$ is  $\Box_a\beta$ for some
    $\beta \in \mathcal{L}$ and $a \in A$, then we also call  $O$ a
    \textit{$\Box_a$-occurrence}. We typically denote an $\alpha$-occurrence by $O[\alpha]$.}
\end{definition}
  
\begin{definition} \normalfont{
A formula $\varphi\in\mathcal{L}$ is an \textit{agent-alternating formula} iff for any $a\in A$ and  any two different $\Box_a$ occurrences $O[\Box_a \alpha]$ and $O[\Box_a \beta]$ such that
$O[\Box_a \alpha] \le O[\Box_a \beta]$,
there is a $b \in A \setminus \{a\}$ and a $\Box_b$-occurrence $O[\Box_b\gamma]$ of $\varphi$ such that 
$O[\Box_a \alpha] \le O[\Box_b \gamma] \le O[\Box_a \beta]$.
In other words, $\varphi$ is agent alternating iff in the tree $\langle OC(\varphi), \le \rangle$, between any two $\Box_a$-occurrences, there is a $\Box_b$-occurrence for some $b \in A \setminus \{a\}$.}
\end{definition}

\begin{example}\normalfont{Assuming $a, b, c$ are different elements in $A$, examples of agent-alternating formulas include:
\[
\Box_a p, ~\Box_a \Box_b p, ~\Box_a \Box_b \Box_a p, ~\Box_a \Box_b \Box_c p, ~\Box_a (p \wedge \Box_b q).
\]
Non-examples include:
\[
\Box_a \Box_a p, ~\Box_a \Box_b \Box_a \Box_a p, ~\Box_a (\Box_b \Box_a p \wedge \Box_a q).
\]}
\end{example}

We now give an equivalent inductive definition of the set of agent-alternating formulas.
\begin{definition}
  \label{def:inductive-def}\normalfont{Define a family $\{ \mathcal{L}_{-a}\}_{a \in A}$ of languages through the following simultaneous induction: 
  $$
  \mathcal{L}_{-a} \ni \varphi ::= 
    p \mid 
    \Box_x\psi  \mid 
    \lnot \varphi \mid 
    (\varphi \land \varphi)
  $$
  where $p \in \mathsf{Prop}$ and $x \in A \setminus \{a\}$ while $\psi \in \mathcal{L}_{-x}$.
  Then the language $\mathcal{L}_{alt}$ is defined inductively by
  $$
  \mathcal{L}_{alt} \ni \varphi ::= 
  p \mid 
  \chi \mid 
  \lnot \varphi \mid 
  (\varphi \land \varphi)
  $$
  where $p \in \mathsf{Prop}$ and $\chi \in \bigcup_{a \in A} \mathcal{L}_{-a}$.}
\end{definition}
Note that $\bigcup_{a\in A}\mathcal{L}_{-a}$ does not cover all of $\mathcal{L}_{alt}$. For example, when $A = \{a, b\}$ with $a \not= b$, $\Box_a p \land \Box_b p$ is in $\mathcal{L}_{alt}$ but not in $\bigcup_{x \in A} \mathcal{L}_{-x}$. 

It is not hard to verify that the two definitions above are equivalent, suggesting that our formal definitions captures the intended intuition. Due to limited space, we omit the  proof of this equivalence, but the idea is simply to examine the parsing trees of formulas. 
\begin{proposition} 
For any $\varphi \in \mathcal{L}$, $\varphi$ is agent alternating iff $\varphi \in \mathcal{L}_{alt}$.
\end{proposition}

\section{Collapsing logics by \texorpdfstring{$\mathcal{L}_{alt}$}{Lalt}}
\label{sec:alt-collapse}
We now investigate which logics are indistinguishable by formulas in $\mathcal{L}_{alt}$. For any normal modal logic $\mathsf{L}$ (defined as a set of formulas in $\mathcal{L}$ satisfying the usual closure properties), let $\mathsf{L}|_{alt}:= \mathsf{L} \cap \mathcal{L}_{alt}$. Then the general question is: for which modal logics $\mathsf{L}$ and $\mathsf{L}'$ are $\mathsf{L}|_{alt}$ and $\mathsf{L}'|_{alt}$ the same? 

More specifically, since we are mainly interested in the introspection axioms $\mathsf{4}$ and $\mathsf{5}$, we focus on the logics appearing in the classic modal logic cube shown in Figure \ref{fig:systems} below.\footnote{Figure \ref{fig:systems} is reproduced from \cite{Garson2018}.} Our main result is that the two shaded areas in Figure \ref{fig:systems} are collapsed in $\mathcal{L}_{alt}$ but no other logics are. To establish this result, we need to first develop bisimulation and unraveling concepts for agent-alternating formulas.  

\begin{notation}
\normalfont{For convenience, we consider $alt$ as an object not in $A$. Also for any set $\mathcal{L}'$ of formulas, $\mathcal{M}, u\equiv_{\mathcal{L}'}\mathcal{N}, v$ means that for all $\varphi \in \mathcal{L}'$, $\mathcal{M}, u \models \varphi$ iff $\mathcal{N}, v \models \varphi$.}  
\end{notation}
\begin{definition}[Agent-alternating bisimulation relation]\normalfont{
  An \textit{agent-alternating bisimulation family} between two models $\mathcal{M}$ and $\mathcal{N}$ is a family of binary relations $\{\bisim_{a}\}_{a \in A\cup\{alt\}}$ between $\mathcal{M}$ and $\mathcal{N}$ such that for every $a \in A\cup\{alt\}$ and every $u \in \mathcal{M}$ and $v \in \mathcal{N}$ such that $u \bisim_{a} v$:
\begin{itemize}
    \item (Atom) for all $p \in \mathsf{Prop}$, $u \in V^{\mathcal{M}}(p)$ iff $v \in V^{\mathcal{N}}(p)$;
    \item (Zig) for all $b \in A \setminus \{a\}$ and  $u' \in R_x^{\mathcal{M}}(u)$, there is
    $v' \in R_x^{\mathcal{N}}(v)$ such that $u' \bisim_{b} v'$;
    \item (Zag) for all $b \in A \setminus \{a\}$ and  $v' \in R_x^{\mathcal{N}}(v)$, there is $u' \in R_x^{\mathcal{M}}(u)$ such that $u' \bisim_{b} v'$. 
\end{itemize}
Then we say $\mathcal{M}, u$ is \textit{agent-alternating bisimilar to $\mathcal{N}, v$} if there is an agent-alternating bisimulation family $\{\bisim_{a}\}_{a \in A \cup \{alt\}}$ between $\mathcal{M}$ and $\mathcal{N}$ such that $u \bisim_{alt} v$.} 
\end{definition}

\begin{lemma}
\label{lem:bisim-inva}
For any models $\mathcal{M}$ and $\mathcal{N}$, agent-alternating bisimulation family $\{\bisim_{a}\}_{a \in A \cup \{alt\}}$ between $\mathcal{M}$ and $\mathcal{N}$, and $a \in A$, if $u \bisim_a v$, then $\mathcal{M}, u \equiv_{\mathcal{L}_{-a}} \mathcal{N}, v$, and if $u \bisim_{alt} v$, then $\mathcal{M}, u \equiv_{\mathcal{L}_{alt}} \mathcal{N}, v$.
\end{lemma}
\begin{proof}
A simple induction on modal depth. 
\end{proof}

\begin{definition}[Agent-alternating unraveling]
\label{def:alt-unraveling}\normalfont{
  Given a model $\mathcal{M} = \langle W^{\mathcal{M}}, \{R^{\mathcal{M}}_a\}_{a\in A}, V^{\mathcal{M}} \rangle$, its \textit{agent-alternating unravelings} are all models of the form $\langle S, \{R_a\}_{a \in A}, V\rangle$ satisfying the following conditions: 
  \begin{itemize}
    \item $S$ is the set of all nonempty finite sequences $s$ of pairs in $(A \cup \{alt\}) \times W^{\mathcal{M}}$ such that
    \begin{itemize}
    \item[(1)] ${s_1 \in \{alt\} \times W^{\mathcal{M}}}$, \item[(2)] ${s_i \in A\times W^{\mathcal{M}}}$ for all $i = 2 \ldots \len(t)$, and 
    \item[(3)] letting $\langle a_i, w_i \rangle = s_i$ for all $i = 1 \ldots \len(s)$, $w_i R^{\mathcal{M}}_{a_{i+1}}w_{i+1}$ and $a_i \not= a_{i+1}$ for all $i = 1 \ldots \len(s) - 1$;
    \end{itemize}
    \item for all $a \in A\cup\{alt\}$ and $s \in S$ such that $s_{\len(s)} = \langle a, w\rangle$, for all $b\in A \setminus \{a\}$, $R_b(s) = \{s + \langle b, w'\rangle \mid w \in R^{\mathcal{M}}_b(w))\}$ (note that this is precisely $\{t \in S \mid s = t_{1 \ldots \len(t)-1}, t_{\len(t)} \in \{b\}\times W^{\mathcal{M}}\}$;
    \item for every $s \in S$ and $p \in \mathsf{Prop}$, $s \in V(p)$ iff $s_{\len(s)} \in {(A \cup \{alt\}) \times V^{\mathcal{M}}(p)}$.
  \end{itemize}
  
  Let $Alt(\mathcal{M})$ denote the set of all agent alternating unraveling of $\mathcal{M}$. Then for every $\mathcal{N} \in Alt(\mathcal{M})$, we define a family of binary relations between $\mathcal{M}$ and $\mathcal{N}$, which we denote as $\{P^{\mathcal{N}}_a\}_{a \in A\cup\{alt\}}$, by
  \[
 u P^{\mathcal{N}}_a s  \iff s_{\len(s)} = \langle a, u\rangle.
  \]}
\end{definition}

\begin{lemma}
\label{lem:alt-unravel-bisim}
  For any model $\mathcal{M}$ and $\mathcal{N} \in Alt(\mathcal{M})$, $\{P^{\mathcal{N}}_a\}_{a \in A\cup\{alt\}}$ is an agent-alternating bisimulation family between $\mathcal{M}$ and $\mathcal{N}$. Consequently, by Lemma \ref{lem:bisim-inva}, for every $w \in \mathcal{M}$, $\mathcal{M}, w \equiv_{\mathcal{L}_{alt}} \mathcal{N}, \langle \langle alt, w\rangle \rangle$.
\end{lemma}
\begin{proof}
Immediate from  Definition \ref{def:alt-unraveling} and the recursive structure of $\mathcal{L}_{alt}$ as defined in Definition \ref{def:inductive-def}.
\end{proof}

Now we can formally state our main result. 
\begin{theorem}
\label{thm:collapsing-in-alt}
Among the systems displayed in Figure \ref{fig:systems}:
\begin{enumerate}
    \item $\mathsf{K}|_{alt} = \mathsf{K4}|_{alt} = \mathsf{K5}|_{alt} = \mathsf{K45}|_{alt} = \mathsf{KB}|_{alt}$;
    \item $\mathsf{KD}|_{alt} = \mathsf{KD4}|_{alt} = \mathsf{KD5}|_{alt} = \mathsf{KD45}|_{alt} = \mathsf{KDB}|_{alt}$;
    \item no other collapse happens when restricting to $\mathcal{L}_{alt}$. 
\end{enumerate}
The results are summarized in Figure \ref{fig:systems-in-alt}, where systems in the same shaded region in Figure \ref{fig:systems} collapse.
\end{theorem}

\begin{proof}
    Combining Proposition \ref{prop:collapse-4-5}, \ref{prop:collapse-B}, and \ref{prop:non-collapsings} below, we have all the collapsing and non-collapsing results in the three layers of Figure \ref{fig:systems}. To see that the three layers do not collapse, it is enough to observe that the axioms $\mathsf{D}$ and $\mathsf{T}$ are in $\mathcal{L}_{alt}$. 
\end{proof}

\begin{figure}[h]
\begin{minipage}{.5\textwidth}
    \centering
    \vspace*{6pt}
    \begin{tikzpicture}[shorten >=1pt,node distance=2cm,on grid,auto, scale=0.5] 
    \node[draw,scale=0.6] at (0, 0, 0) (k4) {$\mathsf{K4}$};
    \node[draw,scale=0.6] at (0, 6, 6) (m) {$\mathsf{T}$};
    \node[draw,scale=0.6] at (0, 3, 6) (kd) {$\mathsf{KD}$};
    \node[draw,scale=0.6] at (7, 0, 6) (kb) {$\mathsf{KB}$};
    \node[draw,scale=0.6] at (7, 3, 6) (kdb) {$\mathsf{KDB}$};
    \node[draw,scale=0.6] at (7, 6, 6) (b) {$\mathsf{B}$};
    \node[draw,scale=0.6] at (0, 0, 6) (k) {$\mathsf{K}$};
    \node[draw,scale=0.6] at (3, 0, 0) (k45) {$\mathsf{K45}$};
    \node[draw,scale=0.6] at (2, 0, 2.5) (k5) {$\mathsf{K5}$};
    \node[draw,scale=0.6] at (6, 0, 0) (kb5) {$\mathsf{KB5}$};
    \node[draw,scale=0.6] at (0, 3, 0) (kd4) {$\mathsf{KD4}$};
    \node[draw,scale=0.6] at (3, 3, 0) (kd45) {$\mathsf{KD45}$};
    \node[draw,scale=0.6] at (1, 2, 0) (kd5) {$\mathsf{KD5}$};
    \node[draw,scale=0.6] at (0, 5.5, 0) (s4) {$\mathsf{S4}$};
    \node[draw,scale=0.6] at (6, 5.5, 0) (s5) {$\mathsf{S5}$};
    \path[->] 
	(k) edge node {} (kd)
	(kd) edge node {} (m)
	(kd) edge node {} (kdb)
	(kd) edge node {} (kd4)
	(kd) edge node {} (kd5)
	(kb) edge node {} (kdb)
	(kdb) edge node {} (b)
	(kd4) edge node {} (kd45)
	(kd4) edge node {} (s4)
	(kd5) edge node {} (kd45)
	(kd45) edge node {} (s5)
	(m) edge node {} (b)
	(m) edge node {} (s4)
	(b) edge node {} (s5)
	(s4) edge node {} (s5)
	(k) edge node {} (kb)
	(kb) edge node {} (kb5)
	(kb5) edge node {} (s5)
	(kb5) edge node {} (s5)
	(k4) edge node {} (k45)
	(k4) edge node {} (kd4)
	(k5) edge node {} (k45)
	(k5) edge node {} (kd5)
	(k45) edge node {} (kb5)
	(k45) edge node {} (kd45)
	(k) edge node {} (k5)
	(k) edge node {} (k4);
	\path[fill=gray,opacity=0.2]
	(0,3,6) -- (0,3,0) -- (3,3,0) -- (7,3,6) -- cycle;
	
	\path[fill=gray,opacity=0.2]
	(0,0,6) -- (0,0,0) -- (3,0,0) -- (7,0,6) -- cycle;
    \end{tikzpicture}
    \vspace*{6pt}
    \caption{Systems of interest}
    \label{fig:systems}
\end{minipage}
\begin{minipage}{.5\textwidth}
    \centering
    \begin{tikzpicture}[shorten >=1pt,node distance=3cm,on grid,auto, scale=0.5] 
   \node[draw,scale=0.6] at (0, 6, 6) (t) {$\mathsf{T}|_{alt}$};
   \node[draw,scale=0.6,align=left] at (0, 3, 6) (kd) 
   {
   $\mathsf{KD}|_{alt}$ \\ 
   $(=\mathsf{KD4}|_{alt},
    \mathsf{KD5}|_{alt},$ \\
    $
    \mathsf{KD45}|_{alt},
    \mathsf{KDB}|_{alt})$
   };
   \node[draw,scale=0.6] at (7, 6, 6) (b) {$\mathsf{B}|_{alt}$};
   \node[draw,scale=0.6,align=left] at (0, 0, 6) (k) 
   {
   $\mathsf{K}|_{alt}$ \\ 
   $(=\mathsf{K4}|_{alt},
    \mathsf{K5}|_{alt},$ \\
    $\mathsf{K45}|_{alt},
    \mathsf{KB}|_{alt})$
   };
   \node[draw,scale=0.6] at (6, 0, 0) (kb5) {$\mathsf{KB5}|_{alt}$};
   \node[draw,scale=0.6] at (0, 5.5, 0) (s4) {$\mathsf{S4}|_{alt}$};
   \node[draw,scale=0.6] at (6, 5.5, 0) (s5) {$\mathsf{S5}|_{alt}$};
   \path[->] 
	(k) edge node {} (kd)
	(k) edge node {} (kb5)
	(kd) edge node {} (t)
	(t) edge node {} (b)
	(t) edge node {} (s4)
	(b) edge node {} (s5)
	(s4) edge node {} (s5)
	(kb5) edge node {} (s5);
\end{tikzpicture}
    \caption{Systems when restricted to $\mathcal{L}_{alt}$}
    \label{fig:systems-in-alt}
\end{minipage}
\end{figure}

\newpage

\begin{proposition}[Collapsing $\mathsf{4}$ and $\mathsf{5}$]
\label{prop:collapse-4-5}
$\mathsf{K}|_{alt} = \mathsf{K45}|_{alt}$ and $\mathsf{KD}|_{alt} = \mathsf{KD45}|_{alt}$
\end{proposition}
\begin{proof}
The right-to-left direction of both equations is trivial. For the left-to-right direction, by completeness, we need only show that for every $\varphi \in \mathcal{L}_{alt}$, if $\varphi$ is satisfied by a pointed model, then it is also satisfied by a pointed model based on a transitive and Euclidean frame. Further, if the first model is based on a serial frame, then the frame of the second model is also serial. So it is enough to show the following: for every pointed model $\mathcal{M}, u$, there exists a pointed model $\mathcal{N}, v$ such that:
\begin{enumerate}
    \item if for every $a \in A$, $R^{\mathcal{M}}_{a}$ is serial, then for every $a \in A$, $R^{\mathcal{N}}_{a}$ is also serial;
    \item for every $a \in A$, $R^{\mathcal{N}}_{a}$ is transitive and Euclidean;
    \item $\mathcal{M}, u \equiv_{\mathcal{L}_{alt}} \mathcal{N}, v$.
\end{enumerate}

Now let $\mathcal{N} = \langle S, \{R_a\}_{a \in A}, V\rangle$ be constructed by adding to the definition of being in $Alt(\mathcal{M})$ as in Definition \ref{def:alt-unraveling} the following:
  \begin{itemize}
    \item for all $a \in A$ and $s \in S$ such that $s_{\len(s)}\in \{a\} \times W^{\mathcal{M}}$, $R_a(s) = R_a(s_{1 \ldots \len(s)-1})$.
  \end{itemize}
This construction is possible because crucially the definition of being an agent-alternating unraveling of $\mathcal{M}$ is silent on what $R_a(s)$ should be when $s$ ends in $\{a\} \times W^{\mathcal{M}}$ for $a \in A$. Also, when $s$ ends in $\{a\} \times W^{\mathcal{M}}$ for some $a \in A$, $\len(s) > 1$ and $s_{1 \ldots \len(s) - 1}$ does not end in $\{a\} \times W^{\mathcal{M}}$, which means that $R_a(s_{1 \ldots \len(s) - 1})$ is defined in Definition \ref{def:alt-unraveling}.

Now we can show that $\mathcal{N}, \langle \langle alt, u\rangle \rangle$ satisfies all the requirements. It is not hard to see that if $\mathcal{M}$ is serial, then so is $\mathcal{N}$. The key observation is that for any $s \in S$, letting $\langle a, w\rangle = s_{len(s)}$, $R_a(s)$ must include $s$, and $R_b(s)$ for any $b \in A \setminus \{s\}$ must be nonempty since $R_b^{\mathcal{M}}(w)$ is nonempty. Hence we are done with (1). To see that for every $a \in A$, $R_a$ is transitive and Euclidean, note that for any $s \in S$, letting $\langle x, w\rangle = s_{\len(s)}$, we have the following:
\begin{itemize}
    \item If $x \not= a$, then for every $t \in R_a(s)$, $t$ ends in $\{a\} \times W^{\mathcal{M}}$, and $t_{1 \ldots \len(t)-1} = s$. This means that our construction above applies to $t$ and $R_a(t) = R_a(s)$.
    \item If $x = a$, then our construction above applies to $s$: letting $s^0 = s_{1 \ldots \len(s) - 1}$, $s^0$ does not end in $\{a\} \times W^{\mathcal{M}}$, and $R_a(s) = R_a(s^0)$ by our definition. Then it is easy to see that for every $t \in R_a(s)$, $t \in R_a(s^0)$, and $t_{1 \ldots \len(t) - 1}$ is also $s^0$. This means that $t$ ends in $\{a\}\times W^{\mathcal{M}}$, and our construction above also applies to $t$. Hence $R_a(t) = R_a(s^0) = R_a(s)$.
\end{itemize}
Adding the above two points together, we have shown that for every $s \in S$ and $t \in R_a(s)$, $R_a(t) = R_a(s)$. This is precisely transitivity plus Euclideanness.

By Lemma \ref{lem:alt-unravel-bisim}, $\mathcal{M}, u \equiv_{\mathcal{L}_{alt}} \mathcal{N}, \langle\langle alt, u\rangle\rangle$ since $\{P_{a}^{\mathcal{N}}\}_{a \in A \cup \{alt\}}$ is an agent-alternating bisimulation family and $u P_{alt}^{\mathcal{N}} \langle \langle alt, u \rangle\rangle$. Thus, all three requirements are satisfied, so we are done. 
\end{proof}

\begin{proposition}[Collapsing $\mathsf{B}$]
\label{prop:collapse-B}
$\mathsf{K}|_{alt} = \mathsf{KB}|_{alt}$ and $\mathsf{KD}|_{alt} = \mathsf{KDB}|_{alt}$
\end{proposition}
\begin{proof}
  Following the strategy of the proof of Proposition \ref{prop:collapse-4-5}, we only need to show that for every pointed model $\mathcal{M}, u$, there exists an agent-alternating unraveling $\mathcal{N} = \langle S, \{R_a\}_{a \in A}, V\rangle$ of $\mathcal{M}$ such that for every $a \in A$, $R_a$ is symmetric. 
  
  Indeed, let $\mathcal{N}$ be the agent-alternating unraveling of $\mathcal{M}$ such that for every $a \in A$ and $s \in S$ such that $s_{\len(s)} \in \{a\} \times W^{\mathcal{M}}$, 
  $R_a(s) = \{s_{1 \ldots \len(s)-1}\}$. Then it is easy to see that for every $a \in A$, $R_a$ is symmetric: for every $s, t \in S$, if $s R_a t$, then we have the following.
  \begin{itemize}
      \item If $s_{\len(s)} \in \{a\} \times W^{\mathcal{M}}$, then $t$ must be $s_{1 \ldots \len(s)-1}$ by our construction. By the definition of unraveling, $t R_a s$.
      \item If $s_{\len(s)} \not\in \{a\} \times W^{\mathcal{M}}$, then $t$ must be $s + \langle a, w\rangle$ for some $w$ such that letting $\langle b, w_0\rangle = s_{\len(s)}$, $w_0 R^{\mathcal{M}}_a w$. Then our construction applies to $t$ and $t R_a s$. 
  \end{itemize}
  Putting the above two points together, $R_a$ is symmetric, so we are done.
\end{proof}

\begin{proposition}[Non-collapsing results]
\label{prop:non-collapsings}
$\mathsf{B}|_{alt} \setminus \mathsf{S4}|_{alt}$, $\mathsf{S4}|_{alt}\setminus\mathsf{B}|_{alt}$, $\mathsf{KB5}|_{alt}\setminus\mathsf{S4}|_{alt}$, $\mathsf{KB5}|_{alt}\setminus\mathsf{B}|_{alt}$ are all nonempty.
\end{proposition}
\begin{proof}
Let $a, b$ be two different elements in $A$. In $\mathsf{B}$ ($=\mathsf{KTB}$), we have the following theorems.
\begin{align*}
&\vdash_{\mathsf{B}} \Box_b\Box_a p \to \Box_a p & &(1) [\mathsf{T}] \\
&\vdash_{\mathsf{B}} \Diamond_a\Box_b\Box_a p \to \Diamond_a\Box_a p & & (2) [\mathsf{RM}, 1] \\
&\vdash_{\mathsf{B}} \Diamond_a\Box_a p \to p & & (3) [\mathsf{B}] \\
&\vdash_{\mathsf{B}} \Diamond_a\Box_b\Box_a p \to p & & (4) [\mathsf{MP}, 2, 3].
\end{align*}
Now the last formula, formula (4), is agent-alternating. However, $\not\vdash_{\mathsf{S4}}(4)$. Using soundness, it is enough to find an $\mathsf{S4}$ model refuting (4). Consider the following model:
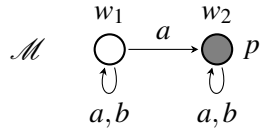
\begin{center}
    \begin{tikzpicture}
        \node[world] (w1) [label=above:$w_1$] {};
        \node[world] (w2) [right=of w1,label=above:$w_2$,fill=gray,label=right:$p$] {};
        \node (M) [left=5mm of w1] {$\mathcal{M}$};
        
        \path (w1) edge[->] node {$a$} (w2);
        \path (w2) edge[loop below] node{$a, b$} (w2);
        \path (w1) edge[loop below] node{$a, b$} (w1);
    \end{tikzpicture}
\end{center}
By focusing on the restriction of $\mathcal{M}$ to $a$ and $b$, respectively, it is easy to see that $\mathcal{M}$ is  based on an $\mathsf{S4}$ frame. Indeed, the accessibility relation for $b$ is even an equivalence relation. Now, $\mathcal{M}, w_1 \models \Diamond_a\Box_b\Box_a p$ since $\mathcal{M}, w_2 \models \Box_b\Box_a p$. Also we have $\mathcal{M}, w_1 \not\models p$. Hence $\mathcal{M}, w_1 \not\models (4)$, and thus $\not\vdash_{\mathsf{S4}}(4)$. This shows that $\mathsf{B}|_{alt} \setminus \mathsf{S4}|_{alt}$ is nonempty. 

In the same spirit, $\Diamond_a\Box_b\Diamond_a p\to \Diamond_a p\in \mathsf{S4}|_{alt}\setminus\mathsf{B}|_{alt}$. The derivation of $\Diamond_a\Box_b\Diamond_a p \to \Diamond_a p$ in $\mathsf{S4}$ is essentially the same as above: using $\mathsf{T}$ we can eliminate the $\Box_b$ in between the two $\Diamond_a$'s. A symmetric countermodel of this formula is as follows.
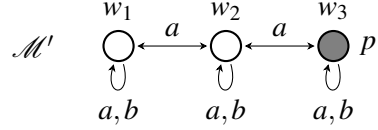
\begin{center}
    \begin{tikzpicture}
        \node[world] (w1) [label=above:$w_1$] {};
        \node[world] (w2) [right=of w1,label=above:$w_2$] {};
        \node[world] (w3) [right=of w2,label=above:$w_3$, fill=gray, label=right:$p$] {};
        \node (M') [left=5mm of w1] {$\mathcal{M}'$};
        
        \path (w1) edge[<->] node {$a$} (w2);
        \path (w2) edge[<->] node {$a$} (w3);
        \path (w3) edge[loop below] node{$a, b$} (w3);
        \path (w2) edge[loop below] node{$a, b$} (w2);
        \path (w1) edge[loop below] node{$a, b$} (w1);
    \end{tikzpicture}
\end{center}

In $\mathsf{KB5}$ we do not have the $\mathsf{T}$ axiom. So $\Diamond_a\Box_b\Box_a p \to p$ and $\Diamond_a\Box_b\Diamond_a p\to \Diamond_a p$ are not in $\mathsf{KB5}$. However, we only need to add $\Box_a\Diamond_b (p \lor \lnot p)$ to the antecedents. Specifically, note that the formula $(\Diamond_b(p \lor \lnot p) \land \Box_b q) \to q$ is in $\mathsf{KB5}$. Hence:
\begin{itemize}
    \item  $\Diamond_a(\Diamond_b(p \lor \lnot p) \land \Box_b\Box_a p) \to p \in \mathsf{KB5}|_{alt} \setminus \mathsf{S4}|_{alt}$;
    \item $\Diamond_a(\Diamond_b(p \lor \lnot p) \land \Box_b\Diamond_a p) \to \Diamond_a p \in \mathsf{KB5}|_{alt} \setminus \mathsf{B}|_{alt}$.
\end{itemize}
Their derivations in $\mathsf{KB5}$ are in the same spirit as above, and $\mathcal{M}$ and $\mathcal{M}'$ can be reused.
\end{proof}

\section{Agent-nonrepeating formulas}\label{sec:nonrepeating}
The above non-collapsing results raise a natural question: is there a smaller fragment defined in the same spirit that also collapses $\mathsf{S5}$ to $\mathsf{T}$? Recall that the non-collapsing results are witnessed by formulas like $\Diamond_a\Box_b\Diamond_a p\to \Diamond_a p$. When $\Box_b$ is factive, agent $a$ is \textit{ipso facto} introspecting since we can eliminate $\Box_b$ by $\mathsf{T}$. In this section we identify a fragment of \textit{agent-nonrepeating formulas} in which this cannot happen and $\mathsf{S5}$ does collapse to $\mathsf{T}$. The key idea is that we need to forbid $\Box_a$ to appear at all in the scope of $\Box_a$. Again, to formalize this idea, we provide an occurrence-based definition and an inductive definition. 

\begin{definition}\normalfont{
A formula $\phi \in \mathcal{L}$ is an \textit{agent-nonrepeating formula} iff for any $a \in A$ and $\Box_a$ occurrence $O[\Box_a \alpha]$, there is no other $\Box_a$ occurrence $O[\Box_a \beta]$ such that $O[\Box_a \beta] \leq O[\Box_a \alpha]$}
\end{definition}
\begin{definition}\normalfont{
Define a family $\{\mathcal{L}_{X}\}_{X \subseteq A}$ of fragments of $\mathcal{L}$ through the following simultaneous induction:
\[
\mathcal{L}_{X} \ni \varphi ::= p \mid \Box_x\psi \mid \lnot \varphi \mid (\varphi \land \varphi)
\]
where $p \in \mathsf{Prop}$ and $x \in X$ while $\psi \in \mathcal{L}_{X \setminus \{x\}}$.} 
\end{definition}

The following equivalence is easily verified.
\begin{proposition}
For any $\varphi \in \mathcal{L}$, $\varphi \in \mathcal{L}_{A}$ iff $\varphi$ is agent nonrepeating.
\end{proposition}

As before, we need a notion of bisimulation appropriate for the fragment.
\begin{definition}\normalfont{
An \textit{agent-nonrepeating bisimulation family} between two models $\mathcal{M}$ and $\mathcal{N}$ is a family of binary relations ${\{\bisim_{X}\}_{X \subseteq A}}$ between $\mathcal{M}$ and $\mathcal{N}$ such that for every $X \subseteq A$ and every $u \in \mathcal{M}$ and $v \in \mathcal{N}$ such that $u \bisim_X v$:
\begin{itemize}
    \item (Atom) for all $p \in \mathsf{Prop}$, $u \in V^{\mathcal{M}}(p)$ iff $v \in V^{\mathcal{N}}(p)$;
    \item (Zig) for all $x \in X$ and  $u' \in R_x^{\mathcal{M}}(u)$, there is
$v' \in R_x^{\mathcal{N}}(v)$ such that $u' \bisim_{X \setminus \{x\}} v'$;
    \item (Zag) for all $x \in X$ and  $v' \in R_x^{\mathcal{N}}(v)$, there is $u' \in R_x^{\mathcal{M}}(u)$ such that $u' \bisim_{X \setminus \{x\}} v'$. 
\end{itemize}
Then we say $\mathcal{M}, u$ is \textit{agent-nonrepeating bisimilar to $\mathcal{N}, v$} if there is an agent-nonrepeating bisimulation family $\{\bisim_{X}\}_{X \subseteq A}$ between $\mathcal{M}$ and $\mathcal{N}$ such that $u \bisim_{A} v$.} 
\end{definition}

\begin{lemma}
\label{lem:nr-invariance}
For any models $\mathcal{M}$ and $\mathcal{N}$,  agent-nonrepeating bisimulation family $\{\bisim_{X}\}_{X \subseteq A}$ between $\mathcal{M}$ and $\mathcal{N}$, and $X \subseteq A$, if $u \bisim_{X} v$, then $\mathcal{M}, u \equiv_{\mathcal{L}_{X}} \mathcal{N}, v$. Hence whenever $\mathcal{M}, u$ is agent-nonrepeating bisimilar to $\mathcal{N}, v$, we have $\mathcal{M}, u \equiv_{L_A} \mathcal{N}, v$.
\end{lemma}

For any logic $\mathsf{L}$, we write $\mathsf{L}|_{nr}$ for $\mathsf{L} \cap \mathcal{L}_{A}$. We can now prove the desired collapse result.
\begin{theorem}
\label{thm:collapsing-S5}
For every reflexive pointed model $\mathcal{M}, w$, there is a partition model $\mathcal{N}, w'$ such that $\mathcal{M}, w$ is agent-nonrepeating bisimilar to $\mathcal{N}, w'$. Consequently, $\mathsf{T}|_{nr} = \mathsf{S5}|_{nr}$. 
\end{theorem}
\begin{proof}
Let $\mathcal{M}$ be a reflexive model. We construct $\mathcal{N} = \langle S, \{R_a\}_{a \in A}, V \rangle$. Let $S$ be the set of all nonempty finite sequences $s$ of pairs in ${\wp(A) \times W^{\mathcal{M}}}$ such that, letting $s = \langle \langle X_i, w_i \rangle \rangle_{i = 1 \ldots \len(s)}$, (1) $X_1 = A$, and (2) for all $i = 1 \ldots \len(s)-1$ and $X_{i+1} \subsetneq X_i$, there is $a\in A$ such that $X_i = X_{i+1} \cup \{a\}$ and $w_iR^{\mathcal{M}}_{a}w_{i+1}$.

To make the rest of the construction easier, we make a few auxiliary definitions. For each $s \in S$, define $LastA(s)$ to be $* \not\in A$ when $\len(s) = 1$ and otherwise the $a \in X \setminus X_0$ with $\langle X, u\rangle = s_{\len(s)}$ and $\langle X_0, u_0\rangle$ when $s_{\len(s) - 1}$. Intuitively, $LastA(s)$ denotes the last accessibility relation used in the sequence $s$.  It is easy to observe from the definition above that for any $s \in S$ with $\langle X, u\rangle = s_{\len(s)}$, $LastA(s) \not\in X$, and moreover when $\len(s) > 1$, $s_{\len(s)-1} = {\langle X \cup LastA(s), u_0\rangle}$ for a $u_0 \in \mathcal{M}$ such that $u_0 R_{LastA(s)}^{\mathcal{M}} u$.

Then for any $a \in A$ and $s \in S$, define $s_{<a}$ to be $s$ if $a \not= LastA(s)$ and otherwise $s_{1 \ldots \len(s)-1}$. Intuitively this is the $a$-predecessor of $s$ in $S$. Now $R_a$ is defined for each $a \in A$ by the condition that $s R_a t$ iff $s_{<a} = t_{<a}$ for all $s, t \in S$. With this definition, it is not hard to compute $R_a(s)$ specifically. For all $s \in S$ such that $\len(s) > 1$, letting $s_{\len(s)} = \langle X, u\rangle$, $s_{\len(s)-1} = \langle X_0, u_0 \rangle$ and $s_0 = s_{1 \ldots \len(s)-1}$, we have the following.
\begin{itemize}
    \item For all $a \in X$, $R_a(s) = \{s + \langle X \setminus \{a\}, v \rangle \mid v \in R_a^{\mathcal{M}}(u)\} \cup \{s\}$.
    \item For the $a \in X_0 \setminus X$ (namely $LastA(s)$), $R_a(s) = \{s_0\} \cup {\{s_0 + \langle X, u' \rangle \mid u_0 R_a^{\mathcal{M}} u' \}}$. 
    \item For $a \in A \setminus X_0$, $R_a(s) = \{s\}$. 
\end{itemize}
For all $s \in S$ such that $\len(s) = 1$, in which case $s = \langle A, u\rangle$ for some $u \in W^{\mathcal{M}}$, we have that $R_a(s) = \{ s + \langle A \setminus \{a\}, v \rangle \mid v \in R_a^{\mathcal{M}}(u)\} \cup \{s\}$. 

The valuation $V$ is defined as usual. For every $s \in S$ and $p \in \mathsf{Prop}$, $s \in V(p)$ iff $s_{\len(s)} \in \wp(A) \times V^{\mathcal{M}}(p)$. That is, $s \in V(p)$ iff the second coordinate of $s_{\len(s)}$ is in $V^{\mathcal{M}}(p)$.

Then there is a natural family $\{\bisim_X\}_{X \in \wp(A)}$ of relations between $\mathcal{M}$ and $\mathcal{N}$ defined by \[ u \bisim_{X} s \iff s_{\len(s)} = \langle Y, u\rangle \textrm{ with } X \subseteq Y.\]

Now we are left with two tasks: to show that $\mathcal{N}$ is a partition model and to show that $\{\bisim_X\}_{X\in\wp(A)}$ is an agent-nonrepeating bisimulation family between $\mathcal{M}$ and $\mathcal{N}$. That $\mathcal{N}$ is a partition model is clear: for any $a \in A$, we defined $R_a$ by an equality condition. Now we show that $\{\bisim_{X}\}_{X \in \wp(A)}$ is an agent-nonrepeating bisimulation family. 
Pick arbitrary $u \in \mathcal{M}, s \in S$, and $X \in \wp(A)$ such that $u \bisim_X s$. By definition, $s_{\len(s)} = \langle Y, u\rangle$ for some $Y \supseteq X$. The (Atom) clause is trivial. For the (Zig) clause, pick an arbitrary $a \in X$ and $v \in R_a^{\mathcal{M}}(u)$. Then we see that $s + \langle Y \setminus \{a\}, v\rangle$ witnesses the requirement, as $s + \langle Y\setminus \{a\}, v\rangle \in R_a(s)$ and $v \bisim_{X \setminus \{a\}} s + \langle Y \setminus \{a\}, v\rangle$ because from $X \subseteq Y$ we have $X \setminus \{a\} \subseteq Y \setminus \{a\}$. For the (Zag) clause, we need to use the fact that $\mathcal{M}$ is reflexive. Picking an arbitrary $a \in X$ and $t \in R_a(s)$, we know that $a \in Y$ and hence there are two cases for $t$:
\begin{itemize}
    \item $t = s$. Then $u$ itself witnesses the requirement, as $u \bisim_{X \setminus \{a\}} s$ and $u R_a^{\mathcal{M}} u$. 
    \item $t = s + \langle Y \setminus \{a\}, v\rangle $ for some $v \in \mathcal{M}$ such that $u R_a^{\mathcal{M}} v$. Then clearly $v$ is witnesses the requirement.
\end{itemize}

In summary, $\mathcal{N}, \langle \langle A, w\rangle\rangle$ is a  pointed partition model, and $\mathcal{M}, w$ is agent-nonrepeating bisimilar to it. Hence we are done. 
\end{proof}

With the help of the above theorem, we obtain the poset of logics in Figure \ref{fig:systems-in-nr} when restricted to $\mathcal{L}_A$. 
\begin{figure}[t!]
\centering
\begin{tikzpicture}[shorten >=1pt,node distance=3cm,on grid,auto, scale=0.6] 
   \node[draw,scale=0.75,align=left] (k) 
   {
   $\mathsf{K}|_{nr}$ \\ 
   $(=\mathsf{K4}|_{nr},
    \mathsf{K5}|_{nr},$ \\
    $\mathsf{K45}|_{nr},
    \mathsf{KB}|_{nr})$
   };
   \node[draw,scale=0.75,align=left, above right=of k] (kd) 
   {
   $\mathsf{KD}|_{nr}$ \\ 
   $(=\mathsf{KD4}|_{nr},
    \mathsf{KD5}|_{nr},$ \\
    $
    \mathsf{KD45}|_{nr},
    \mathsf{KDB}|_{nr})$
   };
   \node[draw,scale=0.75,align=left, below right=of kd] (t) 
   {
   $\mathsf{T}|_{nr}$ \\ 
   $(=\mathsf{S4}|_{nr},
    \mathsf{B}|_{nr},$ \\
    $\mathsf{S5}|_{nr}.$
   };
   \node[draw,scale=0.6, below right=of k]  (kb5) {$\mathsf{KB5}|_{nr}$};
   \path[->] 
	(k) edge node {} (kd)
	(k) edge node {} (kb5)
	(kd) edge node {} (t)
	(kb5) edge node {} (t);
\end{tikzpicture}
\caption{Systems when restricted to $\mathcal{L}_{A}$}
\label{fig:systems-in-nr}
\end{figure}
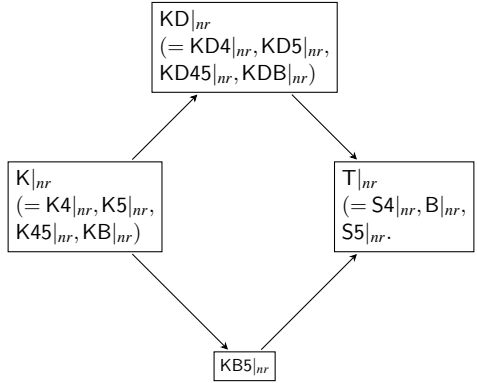

\newpage
\begin{theorem}
Among the systems displayed in Figure \ref{fig:systems}:
\begin{enumerate}
    \item $\mathsf{K}|_{nr} = \mathsf{K4}|_{nr} = \mathsf{K5}|_{nr} = \mathsf{K45}|_{nr} = \mathsf{KB}|_{nr}$;
    \item $\mathsf{KD}|_{nr} = \mathsf{KD4}|_{nr} = \mathsf{KD5}|_{nr} = \mathsf{KD45}|_{nr} = \mathsf{KDB}|_{nr}$;
    \item $\mathsf{T}|_{nr} = \mathsf{S4}|_{nr} = \mathsf{B}|_{nr} = \mathsf{S5}|_{nr}$;
    \item no other collapse happens when restricting to $\mathcal{L}_{A}$. 
\end{enumerate}
The results are summarized in Figure \ref{fig:systems-in-nr}.
\end{theorem}
\begin{proof}
Since $\mathcal{L}_{A} \subseteq \mathcal{L}_{alt}$, all collapsing results in Theorem \ref{thm:collapsing-in-alt} obtain. This covers (1) and (2). Due to Theorem \ref{thm:collapsing-S5}, we have (3). Clearly $\mathsf{KB5}|_{nr} \subseteq \mathsf{S5}|_{nr}$ since the $\mathsf{T}$ axiom is in $\mathcal{L}_{A}$. Hence we are left to show that $\mathsf{KB5}|_{nr}$ is not in $\mathsf{KD}|_{nr}$. The witness is simply $(\Diamond_a (p \lor \lnot p) \land \Box_a p) \to p$.
\end{proof}

\section{Allowing the standard common belief operator?}

\label{sec:commonbelief}
Given its importance in many applications, it is natural to consider adding the standard common belief operator to $\mathcal{L}_{alt}$ and investigate the resulting collapse of logics. In this section, we provide three non-collapsing results for the axioms $\mathsf{4}$ and $\mathsf{5}$, and leave a full investigation with possible collapsing results for future work. Given that $Cp$ expresses a potentially infinitary conjunction of formulas where modalities are compounded in arbitrary order, implicitly $Cp$ is not agent alternating: formulas like $\Box_a\Box_a p$ are part of the definition of $Cp$. Hence it is not surprising that we get many non-collapsing results. Moreover, we face the problem of whether to allow $C$ to be in the scope of or scope over any $\Box_a$ or itself. Again the reason is that if we expand $C\Box_ap$ or $CCp$ syntactically as infinitary formulas, $\Box_a$ will scope over an occurrence of $\Box_a$ immediately. Hence it is not obvious what is the most appropriate definition of an agent-alternating fragment in a language with a common belief operator, and a full investigation would require a hierarchy of fragments, each allowing more interactions between $C$ and other modalities or $C$ itself. Our non-collapsing results about $\mathsf{4}$ also crucially rely on $A$ being finite. We conjecture that the collapsing situation would change radically when $A$ is infinite.

Now let us fix the language and semantics for the common belief operator. 
\begin{definition}\normalfont{
 Let $\mathcal{C}$ be defined by adding new clauses $C\varphi$ and $E\varphi$ to $\mathcal{L}$'s context-free grammars. Semantically, $\mathcal{M}, u \models E\varphi$ iff for all $v \in \mathcal{M}$ such that $u (\bigcup_{a \in A}R^{\mathcal{M}}_a) v$, $\mathcal{M}, v \models \varphi$, and $\mathcal{M}, u \models C\varphi$ iff for all $v \in \mathcal{M}$ such that $u (\bigcup_{a \in A}R^{\mathcal{M}}_a)^+ v$, $\mathcal{M}, v \models \varphi$, where $(\bigcup_{a \in A}R^{\mathcal{M}}_a)^+$ means the transitive closure of the union of relations in $\{R^{\mathcal{M}}_a\}_{a \in A}$. Hence $E\varphi$ formalizes ``everyone believes $\varphi$,'' and $C\varphi$ formalizes ``it is commonly believed that $\varphi$.'' }
\end{definition}
Our logics must expand as well, as we need to add the axioms and rules for the $C$ and $E$ operators. To avoid choosing particular axiomatizations, we define logics directly as validities. For any $\mathsf{L} \subseteq \mathcal{L}$, let $\mathsf{CL}$ denote the set of formulas in $\mathcal{C}$ that are valid on all frames that validates $\mathsf{L}$. For particular axiomatizations, see \cite{Halpern2004infinite}. For our purposes, it is enough to note that for any $\mathsf{L}$, the followings formulas are in $\mathsf{CL}$.
\begin{align*}
    &(Cp \land C(p \to q)) \to Cq && (Ep \land E(p \to q)) \to Eq \\
    &(C(p \to Ep) \land Ep) \to Cp && Ep \to \Box_a p \\
    &Cp \to E(p \land Cp) && \bigwedge_{a \in A} \Box_a p \to Ep \quad (\textrm{when $A$ is finite}).
\end{align*}
Then we can identify at least two $E$-free fragments: one in which $C$ is not allowed to interact with $\Box_a$ but allowed to interact with $C$, and one in which $C$ can appear arbitrarily.
\newcommand{\LC}{\mathcal{L}_{alt}\mathcal{C}^p}
\newcommand{\calt}{\mathcal{C}_{alt}}
\begin{definition}\normalfont{
  Let $\mathcal{C}^p$ be the fragment of formulas in $\mathcal{C}$ with $C$ the only appearing modality. Then let $\mathcal{L}_{alt}\mathcal{C}^p$ be the fragment consisting of Boolean combinations of formulas in $\mathcal{L}_{alt}$ and $\mathcal{C}^p$.}
\end{definition}
\begin{definition}\normalfont{
  Let $\mathcal{C}_{alt}$ and $\mathcal{C}_{-a}$ for any $a \in A$ be defined by adding a new clause $C\varphi$ to $\mathcal{L}_{alt}$ and $\mathcal{L}_{-a}$'s context-free grammars.} 
\end{definition}
For example, $\Box_aC\varphi$ is in $\mathcal{C}_{alt}$ but not in $\LC$. Note that $\LC \subseteq \calt$. Hence for non-collapsing results, using $\LC$ would be stronger. 

Now we present the non-collapsing results. The situation with the $\mathsf{5}$ axiom is relatively simple. Even in the smaller fragment $\LC$ and even with the $\mathsf{D}$ axiom, $\mathsf{5}$ is still important. 
\begin{proposition}
$\mathsf{CK5}\cap\LC$ is not contained in $\mathsf{CKD}\cap\LC$. The formula $\chi_{\mathsf{5}} = (\Diamond_a p \land \Diamond_a \lnot p) \to \widehat{C}(p \land \widehat{C} \lnot p)$ is the witness.
\end{proposition}
\begin{proof}
Clearly the following model proves the claim. All accessibility relations are the same, so we are not labeling the arrows. 
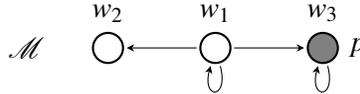
\begin{center}
    \begin{tikzpicture}
        \node[world] (w1) [label=above:$w_2$] {};
        \node[world] (w2) [right=of w1,label=above:$w_1$] {};
        \node[world] (w3) [right=of w2,label=above:$w_3$, fill=gray, label=right:$p$] {};
        \node (M') [left=5mm of w1] {$\mathcal{M}$};
        
        \path (w1) edge[<-] (w2);
        \path (w2) edge[->] (w3);
        \path (w3) edge[loop below]  (w3);
        \path (w2) edge[loop below]  (w2);
    \end{tikzpicture}\vspace{-.3in}
\end{center}
\end{proof}

For the $\mathsf{4}$ axiom we provide two non-collapsing results. First, in $\LC$, $\mathsf{CK4}$ does not collapse to $\mathsf{CK}$ when $A$ is finite. 
\begin{proposition}
$\mathsf{CK4}\cap\LC$ is not contained in $\mathsf{CK}\cap\LC$. The witness is the formula $\chi_\mathsf{4}=(\bigwedge_{x \in A \setminus \{a\}}(\Box_x \bot \land \Box_a\Box_x\bot) \land \Box_a p) \to Cp$.
\end{proposition}
\begin{proof}
The idea is essentially the same as the proof of the next proposition, Proposition \ref{prop:ckd4-separation}.
\end{proof}

The formula in the previous proposition does not separate $\mathsf{CKD4}$ from $\mathsf{CKD}$, as it is trivially valid in $\mathsf{CKD}$ for the reason that $\Box_b\bot$ is inconsistent. Here we provide a formula not in $\LC$ but in $\calt$ that separates $\mathsf{CKD4}$ from $\mathsf{CKD}$, again assuming that $A$ is finite. 
\begin{proposition}
\label{prop:ckd4-separation}
$\mathsf{CKD4}\cap\calt$ is not contained in $\mathsf{CKD}\cap\calt$. The witness is the following formula $\chi_{\mathsf{D4}}$:
\[\left(\bigwedge_{b \in A \setminus \{a\}}(\Box_b p \land C\Box_b p \land \Box_b\Box_a p \land C\Box_b\Box_a p) \land \Box_a p\right) \to C p.\]
\end{proposition}
\begin{proof}
Clearly $\chi_{\mathsf{D4}}$ is in $\mathcal{C}_{alt}$. To see that it is in $\mathsf{CK4}$, recall that the introduction axiom for common belief is 
\[ \left(C(\varphi \to E\varphi) \land E\Box_a \varphi\right) \to C\varphi. \]
Note that $E \Box_a p$ is derivable from the antecedent of $\chi_{\mathsf{D4}}$, as $\Box_b\Box_a p$ for any $b \in A \setminus \{a\}$ is already in the antecedent, and $\Box_a\Box_a p$ follows from $\Box_a p$ by $\mathsf{4}$. Hence we only need $C(p \to Ep)$. It is enough to show $CEp$ and in fact $C\Box_ap$, as for every $b \in A \setminus \{a\}$, $C\Box_b p$ is already in the antecedent of $\chi_{\mathsf{D4}}$. By the $C$-intro axiom again, we only need to derive $E\Box_a p$ and $C(\Box_a p \to E\Box_a p)$.  We already dealt with $E\Box_a p$, so we are left with $C(\Box_a p \to E\Box_a p)$. For any $b \in A \setminus \{a\}$, $C(\Box_a p \to \Box_b\Box_a p)$ follows from $C\Box_b\Box_a p$, which is already in the antecedent of $\chi_{\mathsf{D4}}$. For the case of $C(\Box_a p \to \Box_a \Box_a p)$, we only need to necessitate $\mathsf{4}$. 

Semantically, consider an arbitrary transitive model and a world $u$ in the model. For any $v$ that is reachable from $u$, there are only the following cases.
\begin{itemize}
    \item Only $R_a$ is used. Then $\Box_a p$ being true at $u$ is enough to make $p$ true at $v$, by transitivity.
    \item Only $R_b$ is used for $b \in A \setminus \{a\}$. Since $\Box_b p$ is true at $u$, similarly $p$ is true at $v$.
    \item The last step is in $R_a$, and the last non-$a$ step is $R_b$. Then depending on if there is a step before the last $R_b$ step, $\Box_b\Box_a p$ or $C\Box_b\Box_a p$ being true at $u$  guarantees $p$'s being true at $v$.
    \item The last step is in $R_b$ for some $b\in A \setminus \{a\}$, and $R_b$ is not the only relation used. Then $C\Box_b p$ being true at $u$ guarantees that $p$ is true at v.
\end{itemize}
Hence $p$ is true at $v$ no matter how $v$ is reached from $u$. Thus, $Cp$ is true at $u$. 

Now to see that $\chi_{\mathsf{D4}}$ is not in $\mathsf{CKD}$, consider the following model. 
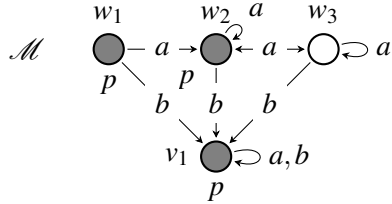
\begin{center}
    \begin{tikzpicture}
        \node[world] (w1) [label=above:$w_1$, label=below:$p$, fill=gray] {};
        \node[world] (w2) [right=of w1,label=above:$w_2$, label=below left:$p$, fill=gray] {};
        \node[world] (w3) [right=of w2,label=above:$w_3$] {};
        \node[world] (v2) [below=of w2,label=left:$v_1$, label=below:$p$, fill=gray] {};
        \node (M') [left=5mm of w1] {$\mathcal{M}$};
        
        \path (w1) edge[->] node[anchor=center, fill=white] {$a$} (w2);
        \path (w2) edge[<->] node[anchor=center, fill=white] {$a$} (w3);
        \path (w2) edge[->, loop,looseness=7,in=30,out=60] node{$a$} (w2);
        \path (w3) edge[loop right] node{$a$} (w3);
        \path (w3) edge[->] node[anchor=center, fill=white]{$b$} (v2);
        \path (v2) edge[<-] node[anchor=center, fill=white]{$b$} (w1);
        \path (w2) edge[->] node[anchor=center, fill=white]{$b$} (v2);
        \path (v2) edge[loop right] node{$a, b$} (v2);
    \end{tikzpicture}
\end{center}
This model $\mathcal{M}$ has all relations serial. Note that at any world, a $b$ step moves you to $v_1$, which makes $p$ and $Cp$ true. Note that $\Box_a p$ is also true at $w_1$. Hence the antecedent is true at $w_1$. But clearly $Cp$ is false at $w_1$, as $w_3$ is reachable but $p$ is false at $w_3$. Thus,  $\chi_{\mathsf{D4}}$ is not in $\mathsf{CKD}$.
\end{proof}

\section{Discussion}\label{sec:discussion}
In the introduction, we suggested that $\mathcal{L}_{alt}$ is sufficient to formalize agents' multi-agent epistemic reasoning in many cases, especially in games. As shown in Appendix~\ref{sec:expressivity}, this claim is substantial if we do not assume both introspection axioms $\mathsf{4}$ and $\mathsf{5}$, for then there is a loss of expressivity in moving from $\mathcal{L}$ to $\mathcal{L}_{alt}$. There remains the question of how widely it is true that $\mathcal{L}_{alt}$ is  sufficient to formalize multi-agent epistemic reasoning. In concrete games, it may well be that there is a brute fact $\varphi$ that is not expressible in $\mathcal{L}_{alt}$, yet for agents to do well in this game, they must reason about $\varphi$. For example, when twins are playing games, there seems to be motivation for them to introspect and reason about themselves. A formal study of this question would complement our work and contribute to answering the question of to what extent introspection axioms matter for multi-agent epistemic reasoning.

In \S~\ref{sec:nonrepeating}, we identified one fragment, the fragment of agent-nonrepeating formulas, with respect to which $\mathsf{S5}$ collapses to $\mathsf{T}$. It is not too hard to see that the expressivity of this fragment is extremely poor. For example, there is a bound on the modal depth of the formulas in this fragment when $A$ is finite. It remains an open question whether there is an expressively more satisfying fragment with a natural syntactic definition that can collapse $\mathsf{S5}$ to $\mathsf{T}$. 

In \S~\ref{sec:commonbelief}, we noted that a full investigation of which fragments of $\mathcal{C}$ collapse which logics is left for future research. In particular, there are two obvious questions. First, when we are separating $\mathsf{CKD4}$ from $\mathsf{CKD}$, the formula we used is in $\calt$ but not in $\LC$. The question here is whether $\LC$ in fact collapses $\mathsf{CLD4}$ to $\mathsf{CKD}$. Second, we did not consider the case where $A$ is infinite. We conjecture that when $A$ is infinite, $\LC$ and perhaps even $\calt$ has the same collapsing power as $\mathcal{L}_{alt}$ does.

\newcommand{\XCY}{\phantom{}^XC^Y}
The main reason for the complexity of the problem with $\mathcal{C}$ is that $C$ is implicitly not agent alternating and hence our unraveling technique does not apply directly. This motivates the formulation of an agent-alternating common belief operator. Indeed, we need many versions of agent-alternating common belief. For any subsets $X$ and $Y$ of $A$, we can define an operator $\phantom{}^XC^Y$ such that $\XCY p$ means that for any nonempty agent alternating finite sequence $l$ of elements in $A$ such that $l_1 \in X$ and $l_{len(l)}\in Y$, $l_1$ believes that $l_2$ believes that $\cdots$ $l_{len(l)}$ believes that $p$. The $X$ and $Y$ are here to make sure that $\XCY$ immediately scopes over and is immediately in the scope of the right modalities. For example, $\Box_a \XCY \Box_b p$ would be agent alternating iff $a \not\in X$ and $b \not\in Y$. It is not hard to see that the techniques in \S~\ref{sec:alt-collapse} are enough to deal with these operators, through a translation to infinitary languages allowing infinite conjunctions, as our agent-alternating bisimulation families preserves truth values of even infinitary formulas. 

Finally, our project can be naturally extended to any extension of the multi-agent doxastic/epistemic language. Natural candidates include languages with dynamic operators, probability operators, or non-standard knowledge operators. The central question to ask in each case is this: what would be a natural agent-alternating fragment or a fragment sufficient for the intended application of those languages, and how does restricting to this fragment affect the landscape of logics? We believe that this type of question will generate interesting results and deepen our understanding of the realm of epistemic logics. 

\bibliographystyle{eptcs}
\bibliography{references}

\begin{thebibliography}{10}
\providecommand{\bibitemdeclare}[2]{}
\providecommand{\surnamestart}{}
\providecommand{\surnameend}{}
\providecommand{\urlprefix}{Available at }
\providecommand{\url}[1]{\texttt{#1}}
\providecommand{\href}[2]{\texttt{#2}}
\providecommand{\urlalt}[2]{\href{#1}{#2}}
\providecommand{\doi}[1]{doi:\urlalt{http://dx.doi.org/#1}{#1}}
\providecommand{\bibinfo}[2]{#2}

\bibitemdeclare{article}{Aumann1976}
\bibitem{Aumann1976}
\bibinfo{author}{Robert~J. \surnamestart Aumann\surnameend}
  (\bibinfo{year}{1976}): \emph{\bibinfo{title}{Agreeing to disagree}}.
\newblock {\sl \bibinfo{journal}{The Annals of Statistics}}
  \bibinfo{volume}{4}(\bibinfo{number}{6}), pp. \bibinfo{pages}{1236--1239},
  \doi{10.1214/aos/1176343654}.

\bibitemdeclare{article}{Aumann1999}
\bibitem{Aumann1999}
\bibinfo{author}{Robert~J. \surnamestart Aumann\surnameend}
  (\bibinfo{year}{1999}): \emph{\bibinfo{title}{Interactive epistemology I:
  Knowledge}}.
\newblock {\sl \bibinfo{journal}{International Journal of Game Theory}}
  \bibinfo{volume}{28}(\bibinfo{number}{3}), pp. \bibinfo{pages}{263--300},
  \doi{10.1007/s001820050111}.

\bibitemdeclare{article}{bernheim1984rationalizable}
\bibitem{bernheim1984rationalizable}
\bibinfo{author}{B.~Douglas \surnamestart Bernheim\surnameend}
  (\bibinfo{year}{1984}): \emph{\bibinfo{title}{Rationalizable strategic
  behavior}}.
\newblock {\sl \bibinfo{journal}{Econometrica}}
  \bibinfo{volume}{52}(\bibinfo{number}{4}), pp. \bibinfo{pages}{1007--1028},
  \doi{10.2307/1911196}.

\bibitemdeclare{article}{bonanno2002modal}
\bibitem{bonanno2002modal}
\bibinfo{author}{Giacomo \surnamestart Bonanno\surnameend}
  (\bibinfo{year}{2002}): \emph{\bibinfo{title}{Modal logic and game theory:
  two alternative approaches}}.
\newblock {\sl \bibinfo{journal}{Risk, Decision and Policy}}
  \bibinfo{volume}{7}(\bibinfo{number}{3}), pp. \bibinfo{pages}{309--324},
  \doi{10.1017/s1357530902000704}.

\bibitemdeclare{book}{Fagin2003}
\bibitem{Fagin2003}
\bibinfo{author}{Ronald \surnamestart Fagin\surnameend},
  \bibinfo{author}{Joseph~Y. \surnamestart Halpern\surnameend},
  \bibinfo{author}{Yoram \surnamestart Moses\surnameend} \&
  \bibinfo{author}{Moshe~Y. \surnamestart Vardi\surnameend}
  (\bibinfo{year}{2003}): \emph{\bibinfo{title}{Reasoning About Knowledge}}.
\newblock \bibinfo{publisher}{MIT Press}.

\bibitemdeclare{unpublished}{Fang2019compile}
\bibitem{Fang2019compile}
\bibinfo{author}{Liangda \surnamestart Fang\surnameend}, \bibinfo{author}{Kewen
  \surnamestart Wang\surnameend}, \bibinfo{author}{Zhe \surnamestart
  Wang\surnameend} \& \bibinfo{author}{Ximing \surnamestart Wen\surnameend}
  (\bibinfo{year}{2018}): \emph{\bibinfo{title}{Knowledge compilation in
  multi-agent epistemic logics}}.
\newblock \bibinfo{note}{{arXiv}: 1806.10561v2}.

\bibitemdeclare{incollection}{Garson2018}
\bibitem{Garson2018}
\bibinfo{author}{James \surnamestart Garson\surnameend} (\bibinfo{year}{2018}):
  \emph{\bibinfo{title}{Modal Logic}}.
\newblock In \bibinfo{editor}{Edward~N. \surnamestart Zalta\surnameend},
  editor: {\sl \bibinfo{booktitle}{The Stanford Encyclopedia of Philosophy}},
  \bibinfo{edition}{fall 2018} edition, \bibinfo{publisher}{Metaphysics
  Research Lab, Stanford University}.

\bibitemdeclare{techreport}{Geanakoplos1989}
\bibitem{Geanakoplos1989}
\bibinfo{author}{John \surnamestart Geanakoplos\surnameend}
  (\bibinfo{year}{1989}): \emph{\bibinfo{title}{Game theory without partitions,
  and applications to speculation and consensus}}.
\newblock \bibinfo{type}{Cowles Foundation Discussion Papers}
  \bibinfo{number}{914}, \bibinfo{institution}{Cowles Foundation for Research
  in Economics, Yale University}.

\bibitemdeclare{phdthesis}{hales2016quantifying}
\bibitem{hales2016quantifying}
\bibinfo{author}{James \surnamestart Hales\surnameend} (\bibinfo{year}{2016}):
  \emph{\bibinfo{title}{Quantifying over epistemic updates}}.
\newblock Ph.D. thesis, \bibinfo{school}{The University of Western Australia}.

\bibitemdeclare{inproceedings}{hales2012refinement}
\bibitem{hales2012refinement}
\bibinfo{author}{James \surnamestart Hales\surnameend}, \bibinfo{author}{Tim
  \surnamestart French\surnameend} \& \bibinfo{author}{Rowan \surnamestart
  Davies\surnameend} (\bibinfo{year}{2012}): \emph{\bibinfo{title}{Refinement
  quantified logics of knowledge and belief for multiple agents}}.
\newblock In: {\sl \bibinfo{booktitle}{Advances in Modal Logic, Volume 9}}, pp.
  \bibinfo{pages}{317--338}.

\bibitemdeclare{article}{halpern2001multi}
\bibitem{halpern2001multi}
\bibinfo{author}{Joseph~Y \surnamestart Halpern\surnameend} \&
  \bibinfo{author}{Gerhard \surnamestart Lakemeyer\surnameend}
  (\bibinfo{year}{2001}): \emph{\bibinfo{title}{Multi-agent only knowing}}.
\newblock {\sl \bibinfo{journal}{Journal of Logic and Computation}}
  \bibinfo{volume}{11}(\bibinfo{number}{1}), pp. \bibinfo{pages}{41--70},
  \doi{10.1093/logcom/11.1.41}.

\bibitemdeclare{article}{Halpern2004infinite}
\bibitem{Halpern2004infinite}
\bibinfo{author}{Joseph~Y. \surnamestart Halpern\surnameend} \&
  \bibinfo{author}{Richard~A. \surnamestart Shore\surnameend}
  (\bibinfo{year}{2004}): \emph{\bibinfo{title}{Reasoning about common
  knowledge with infinitely many agents}}.
\newblock {\sl \bibinfo{journal}{Information and Computation}}
  \bibinfo{volume}{191}(\bibinfo{number}{1}), pp. \bibinfo{pages}{1--40},
  \doi{10.1016/j.ic.2004.01.003}.

\bibitemdeclare{book}{Hintikka1965}
\bibitem{Hintikka1965}
\bibinfo{author}{Jaakko \surnamestart Hintikka\surnameend}
  (\bibinfo{year}{1962}): \emph{\bibinfo{title}{Knowledge and Belief: An
  Introduction to the Logic of the Two Notions}}.
\newblock \bibinfo{publisher}{Cornell University Press}.

\bibitemdeclare{unpublished}{huang2018general}
\bibitem{huang2018general}
\bibinfo{author}{Xiao \surnamestart Huang\surnameend}, \bibinfo{author}{Biqing
  \surnamestart Fang\surnameend}, \bibinfo{author}{Hai \surnamestart
  Wan\surnameend} \& \bibinfo{author}{Yongmei \surnamestart Liu\surnameend}
  (\bibinfo{year}{2018}): \emph{\bibinfo{title}{A general multi-agent epistemic
  planner based on higher-order belief change}}.
\newblock \bibinfo{note}{{arXiv}: 1806.11298v2}.

\bibitemdeclare{article}{Kaneko2002}
\bibitem{Kaneko2002}
\bibinfo{author}{Mamoru \surnamestart Kaneko\surnameend}
  (\bibinfo{year}{2002}): \emph{\bibinfo{title}{Epistemic logics and their game
  theoretic applications: Introduction}}.
\newblock {\sl \bibinfo{journal}{Economic Theory}}
  \bibinfo{volume}{19}(\bibinfo{number}{1}), pp. \bibinfo{pages}{7--62},
  \doi{10.1007/s001990100202}.

\bibitemdeclare{inproceedings}{lakemeyer2012efficient}
\bibitem{lakemeyer2012efficient}
\bibinfo{author}{Gerhard \surnamestart Lakemeyer\surnameend} \&
  \bibinfo{author}{Yves \surnamestart Lesp{\'e}rance\surnameend}
  (\bibinfo{year}{2012}): \emph{\bibinfo{title}{Efficient reasoning in
  multiagent epistemic logics}}.
\newblock In: {\sl \bibinfo{booktitle}{Proceedings of the 20th European
  Conference on Artificial Intelligence}}, \bibinfo{series}{ECAI'12}, pp.
  \bibinfo{pages}{498--503}, \doi{10.3233/978-1-61499-098-7-498}.

\bibitemdeclare{incollection}{Lamarre1994}
\bibitem{Lamarre1994}
\bibinfo{author}{Philippe \surnamestart Lamarre\surnameend} \&
  \bibinfo{author}{Yoav \surnamestart Shoham\surnameend}
  (\bibinfo{year}{1994}): \emph{\bibinfo{title}{Knowledge, certainty, belief,
  and conditionalisation (abbreviated version)}}.
\newblock In \bibinfo{editor}{Jon \surnamestart Doyle\surnameend},
  \bibinfo{editor}{Erik \surnamestart Sandewall\surnameend} \&
  \bibinfo{editor}{Pietro \surnamestart Torasso\surnameend}, editors: {\sl
  \bibinfo{booktitle}{Principles of Knowledge Representation and Reasoning}},
  \bibinfo{series}{The Morgan Kaufmann Series in Representation and Reasoning},
  pp. \bibinfo{pages}{415--424}, \doi{10.1016/b978-1-4832-1452-8.50134-2}.

\bibitemdeclare{article}{Lederman2015}
\bibitem{Lederman2015}
\bibinfo{author}{Harvey \surnamestart Lederman\surnameend}
  (\bibinfo{year}{2015}): \emph{\bibinfo{title}{People with common priors can
  agree to disagree}}.
\newblock {\sl \bibinfo{journal}{Review of Symbolic Logic}}
  \bibinfo{volume}{8}(\bibinfo{number}{1}), pp. \bibinfo{pages}{11--45},
  \doi{10.1017/s1755020314000380}.

\bibitemdeclare{article}{Lenzen1978}
\bibitem{Lenzen1978}
\bibinfo{author}{Wolfgang \surnamestart Lenzen\surnameend}
  (\bibinfo{year}{1978}): \emph{\bibinfo{title}{Recent work in epistemic
  logic}}.
\newblock {\sl \bibinfo{journal}{Acta Philosophica Fennica}}
  \bibinfo{volume}{30}(\bibinfo{number}{2}), pp. \bibinfo{pages}{1--219}.

\bibitemdeclare{inproceedings}{liu2018multi}
\bibitem{liu2018multi}
\bibinfo{author}{Qiang \surnamestart Liu\surnameend} \&
  \bibinfo{author}{Yongmei \surnamestart Liu\surnameend}
  (\bibinfo{year}{2018}): \emph{\bibinfo{title}{Multi-agent epistemic planning
  with common knowledge.}}
\newblock In: {\sl \bibinfo{booktitle}{Proceedings of the 27th International
  Joint Conference on Artificial Intelligence, {IJCAI-18}}}, pp.
  \bibinfo{pages}{1912--1920}, \doi{10.24963/ijcai.2018/264}.

\bibitemdeclare{book}{Meyer1995}
\bibitem{Meyer1995}
\bibinfo{author}{John-Jules~Ch \surnamestart Meyer\surnameend} \&
  \bibinfo{author}{Wiebe van~der \surnamestart Hoek\surnameend}
  (\bibinfo{year}{1995}): \emph{\bibinfo{title}{Epistemic Logic for AI and
  Computer Science}}.
\newblock \bibinfo{publisher}{Cambridge University Press},
  \doi{10.1017/cbo9780511569852}.

\bibitemdeclare{book}{Osborne1994}
\bibitem{Osborne1994}
\bibinfo{author}{Martin~J. \surnamestart Osborne\surnameend} \&
  \bibinfo{author}{Ariel \surnamestart Rubinstein\surnameend}
  (\bibinfo{year}{1994}): \emph{\bibinfo{title}{A Course in Game Theory}}.
\newblock \bibinfo{publisher}{MIT Press}.

\bibitemdeclare{article}{parikh1992levels}
\bibitem{parikh1992levels}
\bibinfo{author}{Rohit \surnamestart Parikh\surnameend} \&
  \bibinfo{author}{Paul \surnamestart Krasucki\surnameend}
  (\bibinfo{year}{1992}): \emph{\bibinfo{title}{Levels of knowledge in
  distributed systems}}.
\newblock {\sl \bibinfo{journal}{Sadhana}}
  \bibinfo{volume}{17}(\bibinfo{number}{1}), pp. \bibinfo{pages}{167--191},
  \doi{10.1007/bf02811342}.

\bibitemdeclare{article}{stalnaker1994evaluation}
\bibitem{stalnaker1994evaluation}
\bibinfo{author}{Robert \surnamestart Stalnaker\surnameend}
  (\bibinfo{year}{1994}): \emph{\bibinfo{title}{On the evaluation of solution
  concepts}}.
\newblock {\sl \bibinfo{journal}{Theory and Decision}}
  \bibinfo{volume}{37}(\bibinfo{number}{1}), pp. \bibinfo{pages}{49--73},
  \doi{10.1007/bf01079205}.

\bibitemdeclare{inproceedings}{Vardi1985}
\bibitem{Vardi1985}
\bibinfo{author}{Moshe~Y. \surnamestart Vardi\surnameend}
  (\bibinfo{year}{1985}): \emph{\bibinfo{title}{A model-theoretic analysis of
  monotonic knowledge}}.
\newblock In: {\sl \bibinfo{booktitle}{Proceedings of the 9th International
  Joint Conference on Artificial Intelligence, {IJCAI-85}}}, pp.
  \bibinfo{pages}{509--512}.

\bibitemdeclare{unpublished}{Vilks1999}
\bibitem{Vilks1999}
\bibinfo{author}{Arnis \surnamestart Vilks\surnameend} (\bibinfo{year}{1999}):
  \emph{\bibinfo{title}{Knowledge of the game, relative rationality, and
  backwards induction without counterfactuals}}.
\newblock \bibinfo{note}{Working Paper no.25, Leipzig Graduate School of
  Management}.

\end{thebibliography}

\appendix
\section{Rationalizability as agent-alternating common belief of rationality}
\label{sec:alternating-rationality}
In this appendix, we sketch a proof that rationality plus agent-alternating common belief of rationality, in which one need not believe that oneself is rational, is enough for all agents to play their rationalizable strategies. This is already implicit in one of the very first definitions of rationalizability by Bernheim \cite{bernheim1984rationalizable}, where he carefully stipulated that agents cannot formulate ``conjectures'' about themselves in systems of beliefs that rationalize their actions. We make this more explicit by using modal logic and Kripke models of games in the style of \cite{stalnaker1994evaluation} and \cite{bonanno2002modal}.

We utilize typical notation in game theory in this appendix. Let $G = \langle \{S_a\}_{a \in A}, \{U_a\}_{a \in A}\rangle$ be a strategic form game: for all $a \in A$, $S_a$ is a finite set and $U_a$ is a function from $\Pi_{a \in A} S_a$ to $\mathbb{R}$, which then naturally extends to $\Delta \Pi_{a \in A} S_a$.

Following \cite{bonanno2002modal}, we first pick a set of distinct proposition letters $\{r_a \mid a \in A\} \subseteq \mathsf{Prop}$. Then a model $\mathcal{M}$ of $G$ is a tuple $\langle W, \{R_a\}_{a \in A}, \{P_a\}_{a \in A}, \{\sigma_a\}_{a \in A}, V\rangle$ satisfying the following properties.
\begin{itemize}
    \item $W$ is a finite set.
    \item For any $a \in A$, $R_a$ is serial binary relation on $W$. 
    \item For any $a \in A$, $P_a$ is a function from $W$ to probability distributions on $W$ such that for any $w \in W$, $P_{a, w}(R_a(w)) = 1$. 
    \item For any $a \in A$, $\sigma_a$ is a function from $W$ to $S_a$.
    \item $V$ is a function from $\mathsf{Prop}$ to $\wp(W)$. 
    \item For any $a \in A$, $w \in V(r_a)$ iff $\sigma_a(w)$ is a best response to what $a$ believes her opponents play. Formally, this condition is
\[   
\forall s_a \in S_a, \sum_{w' \in R_a(w)} P_{a, w}(w') U_a(\sigma_a(w), \sigma_{-a}(w')) \ge \sum_{w' \in R_a(w)} P_{a, w}(w')U_a(s_a, \sigma_{-a}(w')).
\]
\end{itemize}

To formulate agent-alternating common belief of rationality, for each nonempty finite sequence $l$ of elements in $A$, let $\rho_s$ be the formula $\Box_{l_1}\Box_{l_2}\cdots\Box_{l_{len(l)-1}}r_{l_{len(l)}}$. For example, $\rho_{\langle a\rangle} = r_a$ and $\rho_{\langle a, b, a\rangle} = \Box_a\Box_br_a$. Then let $\Gamma$ be the set of $\rho_l$ such that $l$ is agent-alternating: for all $i = 1 \ldots len(l)-1$, $l_i \not= l_{i+1}$. This $\Gamma$ then encodes agent-alternating common belief of rationality. 

With these definitions, it easily follows that if $\mathcal{M}, w \models \Gamma$, then for each $a\in A$, $\sigma_a(w)$ is a strategy that survives the iterated elimination of strictly dominated strategies. To show this, we can simply adapt the proof in \cite{stalnaker1994evaluation}. For each $a \in A$, collect $\sigma_a(w')$ in $Q_a$ for each $w' \in W$ such that there is an agent-alternating path from $w$ to $w'$ with the last move not using $R_a$. Then it is not hard to see that for any $a \in A$ and $q_a \in Q_a$, $q_a$ is not strictly dominated by any strategy in $\Delta Q_a$ assuming any opponent $b \in A \setminus \{a\}$ only plays strategies in $Q_b$. Indeed, given that $q_a$ is in $Q_a$, by definition there is $w' \in W$ and $w_1, w_2, \ldots, w_n$ and $a_1, a_2, \cdots, a_n \not= a$ such that $w_1 = w$, $w_n = w'$, $a_n \not= a$, and for all $i = 1 \ldots n-1$, $w_iR_{i}w_{i+1}$ and $a_i \not= a_{i+1}$. Then for each $w'' \in R_a(w')$, $\sigma_{-a}(w'') \in Q_{-a}$ as $w''$ is reachable from $w$ using $a_1, a_2, \cdots, a_n, a$ which is still an alternating sequence. Note also that $\mathcal{M}, w \models \Gamma$ and in particular $\mathcal{M}, w \models \Box_{a_1}\Box_{a_2} \cdots\Box_{a_{n-1}}r_{a_n}$. Hence $\mathcal{M}, w' \models r_a$, and $q_a$ is the best response to the mixture $m_{-a}$ of $\langle \sigma_{-a}(w'') \mid w'' \in R_a(w')\rangle$ using $P_{a, w'}$, which is a mixture of $Q_{-a}$ since each $\sigma_{-a}(w'') \in Q_{-a}$. Thus $q_a$ is not strictly dominated by any mixture $m_a$ of $Q_a$: if $m_a$ strictly dominates $q_a$ on any $q_{-a} \in Q_{-a}$, then $m_a$ strictly dominates $q_a$ on $m_{-a}$, and then $q_a$ is not the best response to $m_{-a}$, a contradiction. Thus each $Q_a$ survives each stage of elimination. Now for each $a \in A$, $\sigma_a(w) \in Q_a$, since we can use the trivially agent-alternating path $w_1 = w$ and $a_1 = b \in A \setminus \{a\}$ (recall that $|A| > 1$). Hence the strategy $\sigma_q(w)$ played at $w$ is a strategy that survives iterated elimination of strictly dominated strategies. 

We may also express agents' belief that their opponents' actions are independent by proposition letters. Then agent-alternating common belief that agents believe that their opponents' actions are independent can be expressed using modal formulas. To this end, first fix another set $\{r'_a \mid a \in A\} \subseteq \mathsf{Prop}$ of distinct proposition letters such that $\{r'_a \mid a \in A\} \cap \{r_a \mid a \in A\} = \varnothing$. The intended interpretation of $r'_a$ is that $a$ believes that her opponents' actions are independent. Hence we require that for any $a \in A$, $w \in V(r_a')$ iff $P_{a, w}(\{w' \in R_a(w) \mid \sigma_{-a}(w') = s_{-a}\}) = \Pi_{b \in A \setminus \{a\}}P_{a, w}(\{w' \in R_a(w) \mid \sigma_{b}(w') = s_{-a}(b)\})$ for any $s_{-a} \in S_{-a}$. Then similar to the definition of $\rho_{l}$, for any alternating sequence $l$ of elements in $A$, we define $\rho'_{l}$ with the trailing proposition letter being $r'_{l_{len(l)}}$, and we let $\Gamma'$ be the set of all such $\rho'_{l}$'s. With this setup, it is not hard to see, using the same strategy as above, that if $\mathcal{M}, w \models \Gamma \cup \Gamma'$, then for any $a \in A$, $\sigma_a(w)$ is a rationalizable strategy for $a$. 

\section{Expressivity of \texorpdfstring{$\mathcal{L}_{alt}$}{Lalt}}
\label{sec:expressivity}
In this appendix we study the influence of the introspection axioms on the expressivity of $\mathcal{L}_{alt}$. We first define finite agent-alternating bisimulations so we can give a more quantitative analysis of expressivity. 

\begin{definition}\normalfont{
  By induction on $n$, let binary relations $\{\bisim_{-a}^n\}_{a \in A}$ be defined on all pointed models as follows:
  \begin{itemize}
  \item $\mathcal{M}, u \bisim_{-a}^0 \mathcal{N}, v$ iff  $V^{\mathcal{M}}(u) = V^{\mathcal{N}}(v)$.
  \item $\mathcal{M}, u \bisim_{-a}^{n+1} \mathcal{N}, v$ iff
    \begin{itemize}
    \item (Atom) $V^\mathcal{M}(u) = V^{\mathcal{M}}(v)$ and
    \item $\forall b \in A \setminus \{a\}$
      \begin{itemize}
      \item ($b$-zig) $\forall x \in R^{\mathcal{M}}_b (u) \;\exists y \in R^{\mathcal{N}}_b
        (v) \;\mathcal{M}, x \bisim_{-b}^n \mathcal{N}, y$ and 
      \item ($b$-zag) $\forall y \in R^{\mathcal{N}}_b (v)\; \exists x \in
        R^{\mathcal{M}}_b (u)\; \mathcal{M}, x \bisim_{-b}^n \mathcal{N}, y$.
      \end{itemize}
    \end{itemize}
  \end{itemize}
  Then let $\bisim_{alt}^0$ be defined in the same way as $\bisim_{-a}^0$ for any $a \in A$, and define $\bisim_{alt}^{n+1}$ by
    \begin{itemize}
    \item (Atom) $V^\mathcal{M}(u) = V^{\mathcal{M}}(v)$ and
    \item $\forall b \in A$
      \begin{itemize}
      \item ($b$-zig) $\forall x \in R^{\mathcal{M}}_b (u) \;\exists y \in R^{\mathcal{N}}_b
        (v) \;\mathcal{M}, x \bisim_{-b}^n \mathcal{N}, y$ and 
      \item ($b$-zag) $\forall y \in R^{\mathcal{N}}_b (v) \;\exists x \in
        R^{\mathcal{M}}_b (u) \;\mathcal{M}, x \bisim_{-b}^n \mathcal{N}, y$.
      \end{itemize}
    \end{itemize}
  Then for $x \in A \cup \{alt\}$, let $\bisim_{x}^\omega$ be the
  intersection of $\{\bisim_{x}^n\}_{n \in \NN}$.}
\end{definition}

\begin{proposition}
  For all $n \in \NN$ and $a \in A$, $\bisim_{-a}^n$ is an equivalence relation.
  Consequently, $\bisim_{alt}^n$ is an equivalence relation for all $n \in \NN$.
\end{proposition}

\begin{theorem}
\label{thm:alt-bisim-invariance}
  For any $n \in \NN$ and pointed models $\mathcal{M}, u$ and
  $\mathcal{N}, v$:
  \begin{itemize}
  \item $\mathcal{M}, u \equiv_{\mathcal{L}_{-a}^n} \mathcal{N}, _v$ if
    $\mathcal{M}, u \bisim_{-a}^n \mathcal{N}, v$;
  \item $\mathcal{M}, u \equiv_{\mathcal{L}_{alt}^n} \mathcal{N}, v$ if
    $\mathcal{M}, u \bisim_{alt}^n \mathcal{N}, v$;
  \item $\mathcal{M}, u \equiv_{\mathcal{L}_{-a}} \mathcal{N}, v$ if
    $\mathcal{M}, u \bisim_{-a}^{\omega} \mathcal{N}, v$;
  \item $\mathcal{M}, u \equiv_{\mathcal{L}_{alt}} \mathcal{N}, v$ if
    $\mathcal{M}, u \bisim_{alt}^{\omega} \mathcal{N}, v$.
  \end{itemize}
\end{theorem}
Let $\bisim^n$ denote the usual $n$-bisimulation relation and $\bisim^\omega$
the intersection of $\{\bisim^n\}_{n \in \NN}$. 

Now we can present our main results in this appendix. First, with both introspection axioms, there is no loss of expressivity at each modal depth in moving from $\mathcal{L}$ to $\mathcal{L}_{alt}$.
\begin{proposition}
Letting $\mathbf{K45}$ be the class of pointed models where each $R_a$ is transitive and Euclidean, ${\bisim^n} \cap \mathbf{K45}^2 = {\bisim^{n}_{alt}} \cap \mathbf{K45}^2$ for all $n \in \NN$.
\end{proposition}
\begin{proof}
If $|A| = 1$, then ${\bisim^1} \cap \mathbf{K45}^2 = {\bisim^2} \cap \mathbf{K45}^2 = {\bisim^n} \cap \mathbf{K45}^2$ for all $n \ge 1$. It is also easy to see that $\bisim_{alt}^1 \,=\, \bisim^1$. Hence ${\bisim_{alt}^1} \cap \mathbf{K45}^2 = {\bisim^{n}} \cap \mathbf{K45}^2$ for all $n$. The required proposition follows immediately. 

Now we assume $|A| > 1$ and prove the claim by induction on $n$. The case for $n = 0$ is trivial. Now suppose ${\bisim^n} \cap \mathbf{K45}^2 = {\bisim^{n}_{alt}} \cap \mathbf{K45}^2$, and let us show that the claim is true for $n+1$. The left-to-right subset relation is trivial. Hence let us pick two arbitrary pointed models $\mathcal{M}, u$ and $\mathcal{N}, v$ in $\mathbf{K45}$ such that $\mathcal{M}, u \bisim_{alt}^{n+1} \mathcal{N}, u$. Now we need to show that $\mathcal{M}, u \bisim^{n+1} \mathcal{N}, v$. That they have the same atomic valuation is trivial. Now pick an arbitrary $b \in A$ and $u'$ such that $uR_b^{\mathcal{M}}u'$ in $\mathcal{M}$. Our goal is then to find a $v'$ such that $vR_b^{\mathcal{N}}v'$ and $\mathcal{M}, u' \bisim^n \mathcal{N}, v'$. Pick some $a \in A \setminus \{b\}$ (note that we assumed that $|A| > 1$) so that $b \in A \setminus \{a\}$. By the definition of $\bisim_{alt}^{n+1}$, $\mathcal{M}, u \bisim_{-a}^{n+1} \mathcal{N}, v$. Then we obtain a $v'$ such that $v R_b^{\mathcal{N}} v'$ and $\mathcal{M}, u' \bisim_{-b}^{n} \mathcal{N}, v'$. Now we show that this $v'$ is what we need: $\mathcal{M}, u' \bisim^n \mathcal{N}, v'$. By the induction hypothesis, it is enough to show that $\mathcal{M}, u' \bisim_{alt}^n \mathcal{N}, v'$. 

Thus, pick an arbitrary $x \in A$. We need to show that $\mathcal{M}, u' \bisim_{-x}^n \mathcal{N}, v'$. The case for the atomic valuation is again trivial. Now we need to show $y$-zig and $y$-zag for all $y \in A \setminus \{x\}$. When $y \not= b$, they are part of the definition of $\bisim_{-b}^n$, which holds between $\mathcal{M}, u'$ and $\mathcal{N}, v'$. Hence we are left with the case where $y = b$. For $b$-zig, pick an arbitrary $u''$ such that $u'R_b^{\mathcal{M}}u''$. Recall that $u R_b^{\mathcal{M}} u'$ and  $R_b$ is transitive. Hence $u R^{\mathcal{M}}_b u''$. Applying $\mathcal{M}, u \bisim_{-a}^{n+1} \mathcal{N}, v$, we obtain $v''$ such that $vR^{\mathcal{N}}_b v''$ and $\mathcal{M}, u'' \bisim_{-b} \mathcal{N}, v''$. But $R^{\mathcal{N}}_b$ is Euclidean and $vR^{\mathcal{N}}_bv'$ too. Hence $vR^{\mathcal{N}}v''$. Thus this $v''$ witnesses the $b$-zig clause for $\mathcal{M}, u' \bisim_{-x}^n \mathcal{N}, v'$. $b$-zag is shown symmetrically, where the transitivity of $R^{\mathcal{N}}_b$ and the Euclideanness of $R^{\mathcal{M}}_b$ are used. The zag clause for $\mathcal{M}, u \bisim^{n+1} \mathcal{N}, v$ is also shown symmetrically. 
\end{proof}

However, if we consider frame classes corresponding to the modal logic cube as in Figure \ref{fig:systems}, having both introspection properties is necessary.

\begin{proposition}
\label{prop:restricted-non-bisimilar-result}
Letting $\mathbf{S4}$ (resp.~$\mathbf{KD5}$, $\mathbf{B}$) be the class of pointed models where each $R_a$ is reflexive and transitive (resp. serial and Euclidean, reflexive and symmetrical), we have  ${\bisim^{\omega}_{alt}} \cap \mathbf{S4}^2 \not\subseteq {\bisim^2} \cap \mathbf{S4}^2 $,   ${\bisim_{alt}^\omega} \cap \mathbf{KD5}^2 \not\subseteq {\bisim^2} \cap \mathbf{KD5}^2$, and ${\bisim_{alt}^\omega} \cap \mathbf{B}^2 \not\subseteq {\bisim^2} \cap \mathbf{B}^2$.
\end{proposition}
\begin{proof}
The following two models deal with the $\mathbf{S4}$ case. Reflexive loops are omitted. The dashed arrows represent relations for $a$, and the solid ones represent relations for all agents in $A$ beside $a$. 
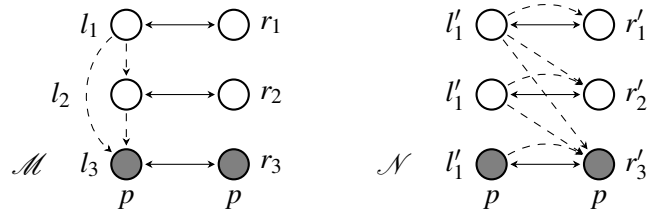
\begin{center}
   \begin{tikzpicture}
    \node[world] (w1) [label=left:$l_1$] {};
    \node[world] (w2) [right=of w1, label=right:$r_1$] {};
   	\node[world] (w3) [below=5mm of w1, label=left:$l_2\quad $] {};
    \node[world] (w4) [right=of w3, label=right:$r_2$] {};
    \node[world] (w5) [below=5mm of w3, fill=gray, label=below:$p$, label=left:$l_3$] {};
    \node[world] (w6) [right=of w5, fill=gray, label=below:$p$, label=right:$r_3$] {};
    
    \node (m) [left=7mm of w5] {$\mathcal{M}$};
        
    \path (w1) edge[<->] node {} (w2);
    \path (w3) edge[<->] node {} (w4);
    \path (w5) edge[<->] node {} (w6);
    \path (w1) edge[->, dashed] node {} (w3);
    \path (w3) edge[->, dashed] node {} (w5);
    \path (w1) edge[->, dashed, bend right=50] node [swap] {} (w5);
        
    \node[world] (w7) [right=30mm of w2, label=left:$l'_1$] {};
    \node[world] (w8) [right=of w7, label=right:$r'_1$] {};
    \node[world] (w9) [below=5mm of w7, label=left:$l'_1$] {};
    \node[world] (w10) [right=of w9, label=right:$r'_2$] {};
	\node[world] (w11) [below=5mm of w9, fill=gray, label=below:$p$, label=left:$l'_1$] {};
    \node[world] (w12) [right=of w11, fill=gray, label=below:$p$, label=right:$r'_3$] {};
        
    \node (n) [left=7mm of w11] {$\mathcal{N}$};
        
    \path (w7) edge[<->] node {} (w8);
    \path (w7) edge[->, dashed, bend left] node [swap] {} (w8);
    \path (w9) edge[<->] node {} (w10);
    \path (w9) edge[->, dashed, bend left] node [swap] {} (w10);
    \path (w11) edge[<->] node [swap] {} (w12);
    \path (w11) edge[->, dashed, bend left] node {} (w12);
    \path (w7) edge[->, dashed] node {} (w10);
    \path (w9) edge[->, dashed] node [swap] {} (w12);
    \path (w7) edge[->, dashed] (w12);
    \end{tikzpicture}
\end{center}
Then we have an agent-alternating family $\{ \bisim_{a}\}_{a \in A\cup \{alt\}}$ of bisimulations. For any $i \in \{1, 2, 3\}$ and  $b \in A \setminus \{a\}$:
\begin{itemize}
    \item $l_i \bisim_a l'_i, r'_i$ and $r_i \bisim_a l'_i, r'_i$;
    \item $l_i \bisim_b l'_i$ and $r_i \bisim_b r'_i$.
\end{itemize}
Essentially the nodes on the same level are connected by $\bisim_a$, and the left column in $\mathcal{M}$ is connected to the left column of $\mathcal{N}$ by $\bisim_b$, and similarly the right column in $\mathcal{M}$ is connected to the right column of $\mathcal{N}$ by $\bisim_b$. Finally it enough to just connect $l_1$ with $l'_1$ by $\bisim_{alt}$. Then it is not hard to check that $\{ \bisim_{a} \}_{a \in A \cup \{alt\}}$ is indeed an agent-alternating bisimulation family. By a simple induction, this clearly implies that $\mathcal{M}, l_1 \bisim_{alt}^{\omega} \mathcal{N}, l'_1$. But of course $\mathcal{M}, l_1 \not\bisim^2 \mathcal{N}, l'_1$ since $\mathcal{M}, l_1 \models \Diamond_a\Diamond_a p $ but $\mathcal{N}, l'_1 \not\models \Diamond_a\Diamond_a p$.

The following two models deal with the $\mathbf{KD5}$ case. 
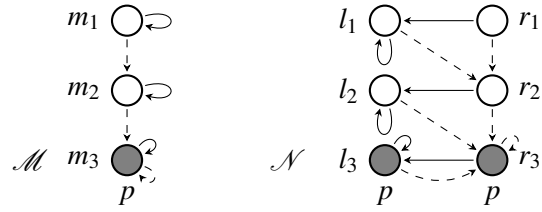
\begin{center}
       \begin{tikzpicture}
        \node[world] (w1) [label=left:$m_1$] {};
        \node[world] (w2) [below=5mm of w1, label=left:$m_2$] {};
       	\node[world] (w3) [below=5mm of w2, label=below:$p$, fill=gray, label=left:$m_3$] {};
        
        \node (m) [left=7mm of w3] {$\mathcal{M}$};
        
        \path (w1) edge[->, dashed] node {} (w2);
        \path (w2) edge[->, dashed] node {} (w3);
        \path (w1) edge[loop right] node {} (w1);
        \path (w2) edge[loop right] node {} (w2);
        \path (w3) edge[->, loop,looseness=7,in=20,out=50] node {} (w3);
        \path (w3) edge[->, dashed, loop,looseness=7,in=-50,out=-20] node{} (w3);       

        \node[world] (w7) [right=3cm of w1, label=left:$l_1$] {};
        \node[world] (w8) [right=of w7, label=right:$r_1$] {};
       	\node[world] (w9) [below=5mm of w7, label=left:$l_2$] {};
        \node[world] (w10) [right=of w9, label=right:$r_2$] {};
	    \node[world] (w11) [below=5mm of w9, label=below:$p$, fill=gray, label=left:$l_3$] {};
        \node[world] (w12) [right=of w11, label=below:$p$, fill=gray, label=right:$r_3$] {};
        
        \node (n) [left=7mm of w11] {$\mathcal{N}$};
        
        \path (w7) edge[<-] node {} (w8);
        \path (w7) edge[loop below] node {} (w7);
        \path (w7) edge[->, dashed] node {} (w10);
        \path (w8) edge[->, dashed] node {} (w10);
        \path (w9) edge[<-] node {} (w10);
        \path (w9) edge[loop below] node {} (w9);
        \path (w9) edge[->, dashed] node {} (w12);
        \path (w10) edge[->, dashed] node {} (w12);
        \path (w11) edge[<-] node [swap] {} (w12);
        \path (w11) edge[->, dashed, bend right] node {} (w12);
        \path (w11) edge[->, loop,looseness=7,in=30,out=60] node{} (w11);
        \path (w12) edge[->, dashed, loop,looseness=7,in=30,out=60] node{} (w12);
    \end{tikzpicture}
\end{center}
This case is easier. For each $i \in \{1, 2, 3\}$, $m_i \bisim_a l_i, r_i$ and $m_i \bisim_b r_i$. Then connecting $m_1$ with $r_1$ by $\bisim_{alt}$, we have an agent-alternating bisimulation family. Hence $\mathcal{M}, m_1 \bisim_{alt} \mathcal{N}, r_1$. However, we have $\mathcal{M}, m_1 \models \Diamond_a \Diamond_a p$ and $\mathcal{N}, r_1 \not\models \Diamond_a \Diamond_a p$. 

Finally, the following two models deal with the $\mathbf{B}$ case. Again, the reflexive loops are omitted from the diagram. The agent-alternating bisimulation family and the formula to refute $\bisim^2$ we need to use are the same as we used in the $\mathbf{S4}$ case.
\begin{center}
    \begin{tikzpicture}
	    \node[world] (w1) [label=left:$l_1$] {};
        \node[world] (w2) [right=of w1, label=right:$r_1$] {};
       	\node[world] (w3) [below=5mm of w1, label=left:$l_2$] {};
        \node[world] (w4) [right=of w3, label=right:$r_2$] {};
	    \node[world] (w5) [below=5mm of w3, label=below:$p$, fill=gray, label=left:$l_3$] {};
        \node[world] (w6) [right=of w5, label=below:$p$, fill=gray, label=right:$r_3$] {};
        
        \node (m) [left=7mm of w5] {$\mathcal{M}$};
        
        \path (w1) edge[-] node {} (w2);
        \path (w3) edge[-] node {} (w4);
        \path (w5) edge[-] node {} (w6);
        \path (w1) edge[-, dashed] node {} (w3);
        \path (w3) edge[-, dashed] node {} (w5);
        
        \node[world] (w7) [right=3cm of w2, label=left:$l^\prime_1$] {};
        \node[world] (w8) [right=of w7, label=right:$r^\prime_1$] {};
       	\node[world] (w9) [below=5mm of w7, label=left:$l^\prime_2$] {};
        \node[world] (w10) [right=of w9, label=right:$r^\prime_2$] {};
	    \node[world] (w11) [below=5mm of w9, label=below:$p$, label=left:$l^\prime_3$, fill=gray] {};
        \node[world] (w12) [right=of w11, label=below:$p$, label=right:$r^\prime_3$,  fill=gray] {};
        
        \node (n) [left=7mm of w11] {$\mathcal{N}$};
        
        \path (w7) edge[-] node {} (w8);
        \path (w9) edge[-] node {} (w10);
        \path (w11) edge[-] node {} (w12);
        \path (w7) edge[-, dashed] node {} (w10);
        \path (w9) edge[-, dashed] node {} (w12);
    \end{tikzpicture}
\end{center}\vspace{-.3in}
\end{proof}

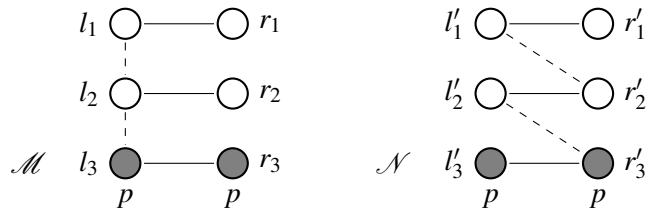
\end{document}